\documentclass[a4paper,USenglish,cleveref,autoref, thm-restate]{lipics-v2021}

\usepackage{todonotes}
\usepackage{amsmath}
\usepackage{amsthm}
\usepackage{amssymb}
\usepackage{xspace}
\usepackage{mathtools}
\usepackage{hyperref}
\usepackage{siunitx}
\usepackage{graphicx}
\usepackage{caption}
\usepackage{csquotes}
\usepackage{hyperref}
\usepackage{booktabs}
\usepackage{xifthen}

\newlength{\punctuationfootlength}
\newcommand{\punctuationfootnote}[2]{#2\settowidth{\punctuationfootlength}%
  {#2}\hspace{-0.5\punctuationfootlength}\footnote{#1}}

\makeatletter
\g@addto@macro\bfseries{\boldmath}
\makeatother

\DeclareMathOperator{\dist}{dist}

\newcommand{\degloc}[1]{{L}_{\deg}(#1)\xspace}
\newcommand{\distloc}[1]{{L}_{\dist}(#1)\xspace}
\newcommand{\loc}[1]{{L}(#1)\xspace}

\newcommand{\EX}[1]{\mathbb E\left[ #1 \right]}
\newcommand{\pnt}[1]{\boldsymbol{#1}}
\newcommand{\geomdist}[3][]{%
  \lVert #2 - #3 \rVert\ifthenelse{\isempty{#1}}{}{_{#1}}}

\newcommand{\cupdot}{\mathbin{\mathaccent\cdot\cup}}

\title{On the External Validity of Average-Case Analyses of Graph
  Algorithms}

\author{Thomas Bläsius}{Karlsruhe Institute of Technology (KIT),
  Germany}{thomas.blaesius@kit.edu}{https://orcid.org/0000-0003-2450-744X}{}

\author{Philipp Fischbeck}{Hasso Plattner Institute (HPI), University of Potsdam, Germany}{philipp.fischbeck@hpi.de}{https://orcid.org/0000-0002-4104-1840}{}

\authorrunning{T.\ Bläsius and P.\ Fischbeck}

\Copyright{Thomas Bläsius and Philipp Fischbeck}



 \relatedversion{An extended abstract of this work was published at ESA 2022~\cite{Exter_Valid_Avera_Analy_ESA2022}.}

\nolinenumbers %uncomment to disable line numbering

\EventShortTitle{PREPRINT}
\ArticleNo{22}
\begin{document}

\maketitle

\begin{abstract}
  The number one criticism of average-case analysis is that we do not
  actually know the probability distribution of real-world inputs.
  Thus, analyzing an algorithm on some random model has no
  implications for practical performance.  At its core, this criticism
  doubts the existence of \emph{external validity}, i.e., it assumes
  that algorithmic behavior on the somewhat simple and clean models
  does not translate beyond the models to practical performance
  real-world input.

  With this paper, we provide a first step towards studying the
  question of external validity systematically.  To this end, we
  evaluate the performance of six graph algorithms on a collection of
  \num{2740} sparse real-world networks depending on two properties;
  the heterogeneity (variance in the degree distribution) and locality
  (tendency of edges to connect vertices that are already close).  We
  compare this with the performance on generated networks with varying
  locality and heterogeneity.  We find that the performance in the
  idealized setting of network models translates surprisingly well to
  real-world networks.  Moreover, heterogeneity and locality appear to
  be the core properties impacting the performance of many graph
  algorithms.
\end{abstract}

\newpage

\section{Introduction}
\label{sec:introduction}

The seminal papers of Cook in 1971~\cite{Compl_Theor_Proce_STOC1971}
and Karp in 1972~\cite{Reduc_Among_Combi_Probl_CCC1972} establish that
many fundamental combinatorial problems are NP-hard, and thus cannot
be solved in polynomial time unless $\mathrm{P} = \mathrm{NP}$.  Since
then, the list of NP-hard problems is growing every year; see the book
by Garey and Johnson~\cite{Compu_Intra_Guide_Theor_other1979} for an
extensive list of problems that were shown to be hard in the early
years of complexity theory.

Though the non-existence of polynomial time algorithms (unless
$\mathrm{P} = \mathrm{NP}$) is major bad news, the concept of
NP-hardness is limited to the worst case.  It thus leaves the
possibility of imperfect algorithms that fail sometimes but run
correctly and in polynomial time on most inputs.  Many algorithms used
today are slow in the worst case but perform well on relevant
instances.  An early attempt to theoretically capture this ``good in
practice'' concept is the \emph{average-case analysis}.  There, one
assumes the input to be randomly generated and then proves a low
failure probability or a good expected
performance\punctuationfootnote{We note that within the scope of this
  paper the term \emph{average case} refers to exactly those
  situations where the input is drawn from some probability
  distribution.  This includes proving bounds that hold with high
  probability (instead of in expectation), which would technically be
  better described as \emph{typical case}.  Moreover, it excludes the
  case of randomized algorithms on deterministic inputs.}.  On that
topic, Karp wrote in 1983~\cite{k-pacoa-83} that
\begin{displayquote}[Karp in 1983]
  One way to validate or compare imperfect algorithms for NP-hard
  combinatorial problems is simply to run them on typical instances
  and see how often they fail. \textelp{} While probabilistic
  assumptions are always open to question, the approach seems to have
  considerable explanatory power \textelp{}.
\end{displayquote}

With this promising starting point, one could have guessed that
average-case analysis is an important pillar of algorithmic research.
However, it currently plays only a minor role in theoretical computer
science.  The core reason for this was summarized by Karp almost forty
years later in 2020 in the Lex Fridman
Podcast\punctuationfootnote{Transcript of the Lex Fridman Podcast
  \#111.  The quote itself starts at 1:39:59.  For the full context,
  start at 1:37:28 (\url{https://youtu.be/KllCrlfLuzs?t=5848}).}.
\begin{displayquote}[Karp in 2020]
  The field tended to be rather lukewarm about accepting these results
  as meaningful because we were making such a simplistic assumption
  about the kinds of graphs that we would be dealing with.  \textelp{}
  After a while I concluded that it didn't have a lot of bite in terms
  of the practical application.
\end{displayquote}

At its core, this describes the issue that an average-case analysis is
lacking \emph{external validity}, i.e., the insights on randomly
generated graphs do not transfer to practical instances.  

The simplistic probabilistic assumption mentioned in the above quotes
is that input graphs are drawn from the Erdős--Rényi
model~\cite{er-rgi-59}, where all edges exist independently at random
with the same probability.  This assumption has the advantages that it
is as unbiased as possible and sufficiently accessible to allow for
mathematical analyses.  However, in its simplicity, it is unable to
capture the rich structural properties present in real-world networks,
leading to the lack of external validity.

That being said, since the beginnings of average-case considerations,
there have been several decades of research in the field of network
science dedicated to understanding and explaining properties observed
in real-world networks.  This in particular includes the analysis of
random network models and the transfer of insights from these models
to real networks; indicating external validity.  Thus, we believe that
it is time to revisit the question of whether average-case analyses of
graph algorithms can have external validity.  With this paper, we
present a first attempt at systematically studying this question.

Before describing our approach and stating our contribution, we want
to give two examples from network science, where the existence of
external validity is generally accepted.

\subparagraph{Examples from Network Science.}

The Barabási--Albert model~\cite{ba-esrn-99} uses the mechanism of
\emph{preferential attachment} to iteratively build a network.  Each
newly added vertex chooses a fixed number of neighbors among the
existing vertices with probabilities proportional to the current
vertex degrees.  This simple mechanism yields \emph{heterogeneous}
networks, i.e., networks with power-law degree distributions where
most vertices have a small degree while few vertices have very high
degree\punctuationfootnote{Barabási and Albert were not the first to
  study a preferential attachment mechanism; see, e.g., Price's
  model~\cite{gener_theor_bibli_other_jour1976}.  However, there is no
  doubt that their highly influential paper~\cite{ba-esrn-99}
  popularized the concept.}.  It is well known that networks generated
by this model are highly artificial, exhibiting properties that are
far from what is observed in real-world networks.  Nonetheless, beyond
the specific model, it is generally accepted that the mechanism of
preferential attachment facilitates power-law distributions.  Thus,
assuming external validity, the Barabási--Albert model can serve as an
explanation of why we regularly observe power-law distributions in
real-world data.  Moreover, whenever we deal with a process involving
preferential attachment, we should not be surprised when seeing a
power-law distribution.

The Watts--Strogatz model~\cite{ws-cdswn-98} first starts with a ring
lattice, i.e., the vertices are distributed uniformly on a circle and
vertex pairs are connected if they are sufficiently close.  This
yields a regular graph with high \emph{locality}, i.e., all
connections are short and we observe many triangles.  Moreover, ring
lattices have high diameter.  The second step of the Watts--Strogatz
model introduces noise by randomly rewiring edges.  This diminishes
locality by replacing local connections with potentially long-range
edges.  Watts and Strogatz demonstrate that only little locality has
to be sacrificed before getting a small-world network with low
diameter.  Again, this model is highly artificial and thus far from
being a good representation for real-world networks.  However, it
seems generally accepted that these observations have implications
beyond the specific model, namely that there is a simple mechanism
that facilitates the small-world property even in networks with mostly
local connections.  Thus, in real-world settings where random
long-range connections are possible, one should not be surprised to
observe the small-world property.

\subparagraph{Contribution.}

We consider algorithms for six different problems that are known to
perform better in practice than the worst-case complexity would
suggest.  We evaluate them on network models that allow for varying
amounts of locality and heterogeneity\punctuationfootnote{For
  non-local networks, we use the Erdős--Rényi and the Chung--Lu model for
  homogeneous and heterogeneous degree distributions, respectively.
  For local networks, we use geometric inhomogeneous random graphs
  (GIRGs), which let us vary the amount of locality and
  heterogeneity.}.  This shows us the impact of these two properties
on the algorithms' performance in the controlled and clean setting of
generated networks.  We compare this with practical performance by
additionally evaluating the algorithms on a collection of \num{2740}
sparse real-world networks.  Our overall findings can be summarized as
follows.  Though the real-world networks yield a less clean and more
noisy picture than the generated networks, the overall dependence of
the performance on locality and heterogeneity coincides with
surprising accuracy for most algorithms.  This indicates that there is
external validity in the sense that if, e.g., increasing locality in
the network models improves performance, we should also expect better
performance for real-world networks with high locality.  Moreover, it
indicates that locality and heterogeneity are the core properties to
impact the performance for many networks.
More specifically, we have the following findings for the different
algorithms.
\begin{itemize}
\item The bidirectional BFS for computing shortest paths in undirected
  networks runs in sublinear time except on homogeneous and local
  networks.  The running times are very similar for the generated and
  real-world networks.
\item The iFUB algorithm~\cite{compu_diame_realw_undir_jour2013} for
  computing the diameter performs well on heterogeneous networks when
  starting it with a vertex of highest degree.  When choosing the
  starting vertex via 4-sweep instead, it additionally performs well
  on networks that are homogeneous and local.  This trend is true for
  generated and real-world networks.
\item The dominance rule, a reduction rule for the vertex cover
  problem, performs well for networks that are sufficiently local or
  heterogeneous.  Again, this trend can be observed for generated and
  real-world networks.
\item The Louvain algorithm~\cite{bgll-fucln-08} for clustering graphs
  requires few iterations for most generated and real-world networks
  independent of locality and heterogeneity.  Though the experiments
  indicate that low locality increases the chance for hard instances,
  the results are inconclusive as locality and heterogeneity do not
  seem to be the main deciding factors.
\item The run time for enumerating all maximal cliques mainly depends
  on the output size, which can be exponential in the worst-case.
  Surprisingly, all generated networks have at most $m$ maximal
  cliques, were $m$ is the number of edges.  Moreover, the number
  decreases for increasing locality.  We make the identical
  observation (not only asymptotically but with the same constant
  factors) on \SI{93}{\%} of the real-world networks.
\item A reduction rule for computing the chromatic number of a graph
  with clique number $\omega$ is to reduce it to its
  $\omega$-core~\cite{Solvi_Maxim_Cliqu_Verte_jour2015}.  It works
  well if the degeneracy is low compared to $\omega$.  For generated
  and real-world networks, the clique number and the degeneracy behave
  almost identical, both increasing for higher locality and
  heterogeneity.  The reduction rule itself has decent performance on
  very heterogeneous networks.  On less heterogeneous networks, it
  performs better for higher localities.  In addition, we observe an
  interesting threshold behavior in the average degree, depending on
  the locality, for the generated networks.
\end{itemize}

Our insights for the specific algorithms are interesting in their own
right, independent of the question of external validity.  Moreover,
our experiments led to several interesting findings that are beyond
the core scope of this paper.  These can be found in the appendix.

In Section~\ref{sec:basic-defin-heter}, we introduce some basic
notation and formally define measures for heterogeneity and locality.
In Section~\ref{sec:data-set}, we describe the set of real-world and
generated networks we use in our experiments.
Section~\ref{sec:comp-betw-real-world-and-models} compares generated
and real-world networks for the different algorithms.  Related work as
well as our insights specific to the algorithms are discussed in this
section.  In Section~\ref{sec:disc-concl} we conclude with a
discussion of our overall results.

Our source code is available on
GitHub\footnote{\url{https://github.com/thobl/external-validity}}.  It
is additionally archived at
Zenodo\footnote{\url{https://doi.org/10.5281/zenodo.8058432}},
together with a docker image for easier reproducibility.  The latter
repository additionally includes the real-world network data set as
well as all generated data (networks and statistics).

\section{Basic Definitions, Heterogeneity, and Locality}
\label{sec:basic-defin-heter}

Let $G = (V, E)$ be a graph.  Throughout the paper, we denote the
number of vertices and edges with $n = |V|$ and $m = |E|$.  For
$v \in V$, let $N(v) = \{u \mid \{u, v\} \in E \}$ be the
\emph{neighborhood} of $v$, and let $\deg(v) = |N(v)|$ be the
\emph{degree} of $v$.  Additionally, $N[v] = N(v) \cup \{v\}$ is the
\emph{closed neighborhood} of $v$.  An edge $e \in E$ is a
\emph{bridge} if $G - e$ is disconnected, where $G - e$ denotes the
subgraph induced by $E \setminus \{e\}$.

\subsection{Heterogeneity}

We define the \emph{heterogeneity} of a graph as the logarithm (base
10) of the coefficient of variation of its degree distribution.  To
make this specific, let $\mu = \frac{1}{n} \sum_{v \in V} \deg(v)$ be
the average degree of $G = (V, E)$, and let
$\sigma^2 = \frac{1}{n} \sum_{v \in V} (\deg(v) - \mu)^2$ be the
variance.  Then, the \emph{coefficient of variation} is
$\sigma / \mu$, i.e., the standard deviation relative to the mean.
Thus, the heterogeneity is $\log_{10}(\sigma / \mu)$.  The resulting
distribution of heterogeneity values is shown in
Figure~\ref{fig:locality_heterogeneity_density}. The figure includes thresholds for extreme heterogeneity values that we use to filter real-world graphs in our visualizations; see Section~\ref{sec:comp-betw-real-world-and-models} for details.

\subsection{Locality}
\label{sec:definition-locality}

We define locality as a combination of two different notions of
locality on edges.  The degree locality of an edge is high if its
endpoints have many common neighbors.  This is similar to the commonly
known local clustering coefficient of vertices.  Our second measure,
the distance locality, captures the locality of edges that do not have
common neighbors.  In the following, we first introduce these
parameters.  The subsequent discussion helps to interpret them and
justifies our choices.  For the distribution of the different locality
values over the networks see
Figure~\ref{fig:locality_heterogeneity_density}.

\begin{figure}
  \centering
  \includegraphics{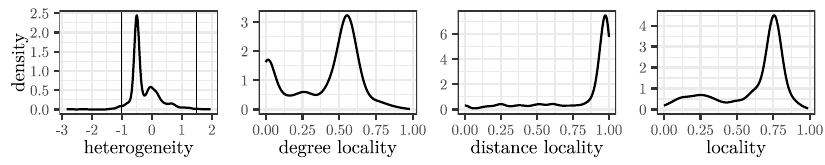}
  \caption{The density (kernel density estimation) of heterogeneity,
    degree locality, distance locality, and locality of the networks
    in our data set of real-world networks.}
  \label{fig:locality_heterogeneity_density}
\end{figure}

\subparagraph{Degree Locality.}

For $u, v \in V$, let $\deg(u, v) = |N(u) \cap N(v)|$ be the
\emph{common degree} of $u$ and $v$.  For a non-bridge edge
$\{u, v\} \in E$ the \emph{degree locality} is defined as
\begin{equation*}
  \degloc{\{u, v\}} = \frac{\deg(u, v)}{\min(\deg(u), \deg(v)) - 1}. 
\end{equation*}
The $-1$ in the denominator accounts for the fact that $u$ is always
in $v$'s neighborhood but never in the common neighborhood.  With
this, we get $\degloc{\{u, v\}} \in [0, 1]$ and
$\degloc{\{u, v\}} = 0$ if and only if the neighborhoods of $u$ and
$v$ are disjoint.  Moreover, $\degloc{\{u, v\}} = 1$ if and only if
$u$ is only connected to $v$ and to neighbors of $v$ or vice versa.
Essentially, $\degloc{\{u, v\}}$ measures in how many triangles
$\{u, v\}$ appears.  Note that the denominator is never $0$ as we
require the edge to not be a bridge and thus $u$ and $v$ both have
degree at least~$2$.

The \emph{degree locality} $\degloc{G}$ of $G$ is the average degree
locality over all non-bridge edges.

\subparagraph{Distance Locality.}

For $u, v \in V$, let $\dist(u, v)$ be the distance between $u$ and $v$
in $G$.  If $\dist(u, v) = 1$, i.e., $\{u, v\}$ is an edge, we are
additionally interested in the detour we have to make when not
allowing to use the direct edge $\{u, v\}$.  To this end, we define
the \emph{detour distance} $\dist^+(u, v)$ to be the distance between
$u$ and $v$ in $G - \{u, v\}$.

For a set of vertex pairs $P \subseteq {V \choose 2}$ we define
$\dist(P)$ to be the \emph{average distance of $P$}, i.e.,
\begin{equation*}
  \dist(P) = \frac{1}{|P|} \sum_{\{u, v\} \in P} \dist(u, v).
\end{equation*}
Analogously, we define $\dist^+(P)$ to be the average detour distance
of $P$.

Let $\overline{E} = {V \choose 2} \setminus E$ be the \emph{non-edges}
of $G$ and assume that $\overline{E} \not= \emptyset$.  Note that
$\dist(\overline{E}) \ge 2$, as non-adjacent vertex pairs have
distance at least~$2$.  Assume for now that $\dist(\overline{E}) > 2$.
For a non-bridge edge $\{u, v\} \in E$, we define the \emph{distance
  locality} as
\begin{equation*}
  \distloc{\{u, v\}} =  1 -
  \frac{\dist^+(u, v) - 2} {\dist(\overline E) - 2}.
\end{equation*}
Note that the numerator is $0$ if and only if $u$ and $v$ have a
common neighbor, which yields $\distloc{\{u, v\}} = 1$.  Moreover, we
have $\distloc{\{u, v\}} = 0$ if and only if the detour distance
between $u$ and $v$ equals the average distance between non-adjacent
vertex pairs in $G$.  Finally, $\distloc{\{u, v\}}$ can be negative if
$\{u, v\}$ connects a vertex pair that is otherwise more distant than
one would expect for a non-adjacent vertex pair.  Thus, the distance
locality essentially measures how short the edge $\{u, v\}$ is
compared to the average distance in the graph.  In the special case of
$\dist(\overline{E}) = 2$, we define $\distloc{\{u, v\}} = 0$.

The \emph{distance locality} $\distloc{G}$ of $G$ is the maximum of
$0$ and the average distance locality over all non-bridge edges.

\subparagraph{Locality.}

The \emph{locality} $\loc{G}$ of $G$ is the average of the degree and
the distance locality, i.e.,
$\loc{G} = (\degloc{G} + \distloc{G}) / 2$.

To interpret this parameter, let
$\loc{e} = \frac{1}{2}(\degloc{e} + \distloc{e})$ be the \emph{edge
  locality} of a non-bridge edge $e = \{u, v\} \in E$.  Note that
$\loc{G}$ is basically\footnote{This is true unless the average over
  all distance localities is negative, in which case we capped the
  distance locality at~$0$.  See
  Section~\ref{sec:limitations-locality-details} for a detailed
  discussion.} the average of all edge localities.  Observe that
$\degloc{e} > 0$ and $\distloc{e} = 1$ if $u$ and $v$ have a common
neighbor.  Otherwise $\degloc{e} = 0$ and $\distloc{e} < 1$.  Thus, we
in particular get the following regimes for $\loc{e}$.

\begin{itemize}
\item $\loc{e} \in \left(\frac{1}{2}, 1\right]$ if $u$ and $v$ have a
  common neighbor.  The more common neighbors $u$ and $v$ have, the
  higher $\loc{e}$.
\item $\loc{e} \in \left(0, \frac{1}{2}\right)$ if $u$ and $v$ have no
  common neighbor but are closer in $G - e$ than the average
  non-adjacent vertex pair in $G$.  The closer $u$ and $v$ are, the
  higher $\loc{e}$.
\end{itemize}

\subparagraph{Discussion \& Comparison to the Clustering Coefficient.}

The degree locality is closely related to the commonly known local
clustering coefficient.  Note that the degree locality (like the
clustering coefficient) only cares for triangles.  Though this is
desirable in some cases, it does not provide a good separation between
graphs with few triangles.  An extreme case are bipartite graphs that
have no triangles and thus degree locality and clustering
coefficient~\num{0}.  However, we would regard, e.g., grids as highly
local.  The distance locality solves this issue by essentially
defining a measure of locality that distinguishes between graphs of
low degree locality.
Figure~\ref{fig:locality_heterogeneity_clustering} shows a comparison
of our locality values to the local clustering coefficient.  Note how,
for high values, the locality behaves as the clustering coefficient
(scaled by a factor of \num{2}), while it provides additional
separation for networks with clustering close to~\num{0}.  See
section~\ref{sec:comp-with-clust-details} for an extended discussion.

\begin{figure}
  \centering
  \includegraphics{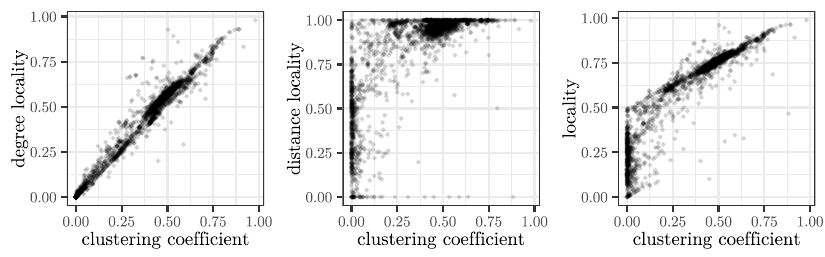}
  \caption{Comparison of the average local clustering coefficient to
    the degree locality (left), the distance locality (center), and
    the locality (right).  Each dot represents one network from our
    data set of real-world networks.}
  \label{fig:locality_heterogeneity_clustering}
\end{figure}

\subparagraph{Limitations.}

We note that our locality definition has some limitations, e.g., it is
undefined for trees as trees have no non-bridge edges.  However, these
limitations do not pose an issue in the context of this paper; see
Section~\ref{sec:limitations-locality-details} for more details.

\subparagraph{Computing the Locality.}

Computing the distance locality efficiently is not straight-forward.
In Section~\ref{sec:comp-edge-local-details}, we show how to compute
it by computing two values for a given graph: the average distance
between all vertex pairs, and the average detour distance over all
non-bridge edges.  The latter can be computed exactly in reasonable
time using some insights from Section~\ref{sec:bidirectional-search}
(see Section~\ref{sec:comp-edge-local-details} for details).

For the average distance between all vertex pairs, we use an
approximation.  Specifically we implemented the algorithm by Chechik,
Cohen, and Kaplan~\cite{Avera_Dista_Queri_throu_APPROX2015}.  We
observed that there is a simple way to significantly improve its
approximation ratio by conditioning on random choices that have been
made earlier in the algorithm.  Additionally, we compared the
algorithm with the straight-forward method of computing the shortest
path between uniformly sampled vertex pairs.  Our experiments indicate
that the preferable method depends on the locality and heterogeneity
of a network.  Though these results are interesting in their own
right, they are somewhat beyond the scope of this paper and thus
deferred to Section~\ref{sec:average-distance}.

\section{Networks}
\label{sec:data-set}

\subsection{Real-World Networks}

We use \num{2740} graphs from Network
Repository~\cite{Netwo_Data_Repos_with_AAAI2015}, which we selected as
follows.  We started with all networks with at most \SI{1}{M} edges
and reduced each network to its largest connected component.  We
removed multi-edges and self-loops and ignored weights or edge
directions.  We tested the resulting connected, simple, undirected,
and unweighted graphs for isomorphism and kept only one copy for each
group of isomorphic networks.  This resulted in \num{2977} networks.
From this we removed \num{237} networks for different reasons.
\begin{itemize}
\item We removed \num{111} networks that were randomly generated and
  thus do not count as real-world networks.\footnote{Finding generated
    networks was done manually by searching for suspicious naming
    patterns or graph properties.  For each candidate, we checked its
    source to verify that it is a generated network.  Though we
    checked thoroughly, there are probably a few random networks
    hiding among the real-world networks.}
\item We removed \num{2} trees.  They are not very interesting, and
  locality is not defined for trees.
\item We removed \num{124} graphs with density at least \SI{10}{\%}.
  Our focus lies on sparse graphs and network models for sparse
  graphs.  Thus, dense graphs are out of scope.
\end{itemize}

\subsection{Random Networks}
\label{sec:random-networks}

We use three random graph models to generate networks; the
Erdős--Rényi model~\cite{er-rgi-59} (non-local and homogeneous), the
Chung--Lu model~\cite{cl-adrgged-02,cl-ccrgg-02} (non-local and
varying heterogeneity), and the GIRG
model~\cite{Geome_inhom_rando_graph_jour2019} (varying locality and
heterogeneity).  For the latter, we use the efficient implementation
in~\cite{Effic_Gener_Geome_Inhom_ESA2019}.  In the following we
briefly define the models, discuss the model choice, and specify what
networks we generate.

\subparagraph{Network Model Definitions.}

Given $n$ and $m$, the \emph{Erdős--Rényi model} draws a graph
uniformly at random among all graphs with $n$ vertices and $m$ edges.

Given $n$ weights $w_1, \dots, w_n$ with $W = \sum w_v$, the
\emph{Chung--Lu model} generates a graph with $n$ vertices by creating
an edge between the $u$th and $v$th vertex with probability
\begin{equation*}
  p_{u, v} = \min\left\{\frac{w_u w_v}{W}, 1\right\}.
\end{equation*}
The expected degree of $v$ is then roughly proportional to $w_v$.  We
use weights that follow a power law.  Specifically, for a constant $c$
and a power-law exponent $\beta > 2$, we choose
\begin{equation*}
  w_v = c\cdot v^{-\frac{1}{\beta - 1}}.
\end{equation*}
Thus, the parameters are the number of vertices $n$, the expected
average degree (controlled via $c$), and the power-law exponent
$\beta$.  The latter controls the heterogeneity of the network:
heterogeneity is high for small $\beta$ and for $\beta \to \infty$ we
get uniform weights.

The \emph{GIRG model} (geometric inhomogeneous random graphs) augments
the Chung--Lu model with an underlying geometry.  In a first step,
each vertex $v$ is mapped to a random position $\pnt{v}$ in some
$d$-dimensional ground space.  Let $\geomdist{\pnt{u}}{\pnt{v}}$
denote the distance between vertices $u$ and $v$ in that space.  Then,
for a so-called \emph{temperature} $T \in (0, 1)$, the vertices $u$
and $v$ are connected with probability
\begin{equation*}
  p_{u, v} =
  \min\left\{\left(\frac{1}{\geomdist{\pnt{u}}{\pnt{v}}^d}\cdot\frac{w_u
        w_v}{W}\right)^{\frac{1}{T}}, 1\right\}.
\end{equation*}
Additionally, for $T=0$, a threshold variant is obtained, with
$p_{u, v} = 1$ if $\geomdist{\pnt{u}}{\pnt{v}}^d \leq w_u v_w / W$ and
$p_{u, v} = 0$ otherwise.  For the weights, we use the same power-law
weights as for the Chung--Lu model.  As ground space, we usually use a
2-dimensional torus $\mathbb T^2 = [0, 1] \times [0, 1]$ with maximum
norm, which is basically just a unit square with distances wrapping
around in both dimensions.  More precisely, for
$\pnt{u} = (x_u, y_u) \in \mathbb T$ and
$\pnt{v} = (x_v, y_v) \in \mathbb T$, we have
$\geomdist{\pnt{u}}{\pnt{v}} = \max\{\min\{|x_u - x_v|, 1 - |x_u -
x_v|\},\min\{|y_u - y_v|, 1 - |y_u - y_v|\}\}$.  With this, the
parameters that remain to be chosen for the GIRG model are the same as
for the Chung--Lu model plus the temperature $T$, which mainly
controls the locality.

\subparagraph{Discussion of the Network Models.}

We note that there is a plethora of network models out there.  So, why
did we choose these particular models?  To answer this, consider two
different use-cases for network models.  First, to explain why and how
certain properties emerge in networks assuming some mechanism existing
in the real world.  Secondly, to draw conclusions from the existence
of certain properties in a network.

An example of the first perspective is the Barabási--Albert
model~\cite{ba-esrn-99}, where the simple and believable mechanism of
preferential attachment yields a power-law degree distribution.

The second perspective is what we are concerned with in this paper.
We do not aim at explaining how locality or heterogeneity emerge in
networks.  Instead we want to study the implications of these
properties on algorithmic performance.  For this purpose, we believe
that it is crucial for the model to be maximally unbiased beyond the
explicitly assumed properties.  Ideally, one uses the maximum entropy
model with respect to the given constraints.

With this in mind, note that the Erdős--Rényi model draws graphs
uniformly at random, given the number of vertices and edges.  This is
the maximum entropy model and thus maximally unbiased with the
constrains that the number of vertices and edges are given.

For graphs whose degree distribution follows a power law, one could
use the above mentioned Barabási--Albert model.  However, this model
is biased in all sorts of ways that go beyond assuming a power-law
degree distribution (e.g., the resulting graphs are $k$-degenerate for
average degree $2k$).  The maximum entropy model for graphs with
power-law distribution is the (soft) configuration
model~\cite{vlk-smerggpdd-18}.  The Chung--Lu model we
use is close to the soft configuration model but much easier to work
with.  Thus, it is a good compromise between being unbiased beyond the
assumptions we want to make and being able to efficiently generate
networks or conduct theoretical analyses.

Finally, the GIRG model enhances the Chung--Lu model with locality by
assuming an underlying geometry.  Though the model is set up to be as
unbiased as possible beyond this assumption, we are not aware of a
formal proof for this.  Moreover, it is up for debate whether a
geometry is the best way to introduce locality.  Instead, one could,
e.g., assume a fixed set of communities like in the stochastic block
model~\cite{hll-sbfs-83}.  However, assigning vertices to geometric
positions seems like the most straight-forward way of introducing a
smooth notion of similarity.  Thus, at the current state of research,
we believe that assuming an underlying geometry is the best way to
achieve locality for our purpose.  As an added bonus, the amount of
locality can be controlled with just a single parameter (the
temperature) and there is an efficient
generator~\cite{Effic_Gener_Geome_Inhom_ESA2019} readily available.

\subparagraph{Generated Networks \& Parameter Choices.}

For all models, we generate networks with $n = \SI{50}{k}$ vertices
and (expected) average degree \num{10}.  For the power-law exponent
$\beta$ in Chung--Lu graphs and GIRGs as well as for the temperature
$T$ in GIRGs we chose the values shown in
Figure~\ref{fig:params-of-generated-networks}, yielding a rather
uniform distribution of heterogeneity and locality.

\begin{figure}[t]
  \centering
  \includegraphics{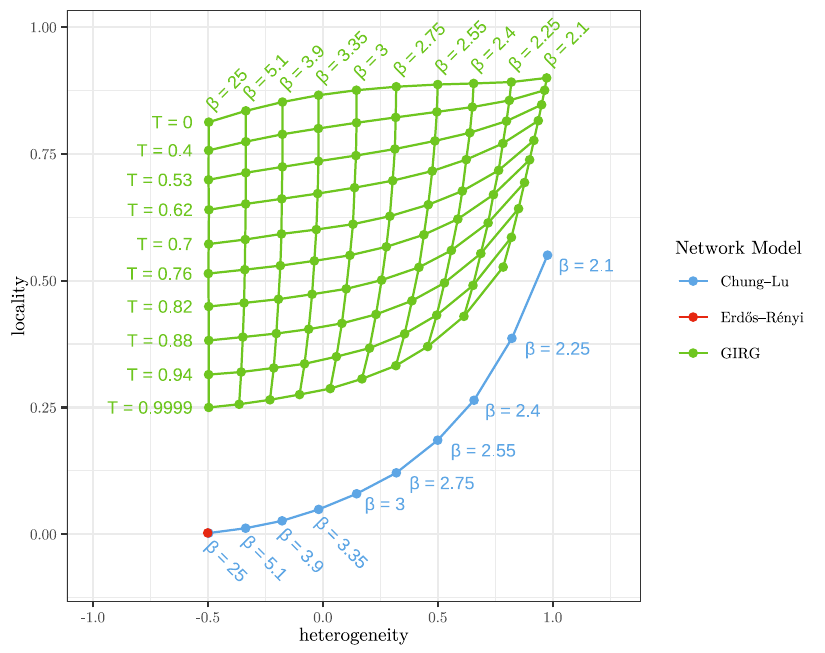}
  \caption{Heterogeneity and locality of the generated networks from
    the different models.  Each point is the average of five samples
    with the given parameter configuration.}
  \label{fig:params-of-generated-networks}
\end{figure}

Note that we have ten different values for each of the parameters
$\beta$ and $T$, which results in \num{100} different parameter
configurations for GIRGs, ten configurations for Chung--Lu graphs, and
a single configuration for Erdős--Rényi graphs.  For each of these
\num{111} configurations, we generated five networks.  In the plots,
we always use a single dot for each configuration representing the
average over five individual networks.

As for the real-world networks, we reduce each generated graph to its
largest connected component.  With the above parameter choices, this
does not change the size of the networks by too much; see
Section~\ref{sec:giant-component}.

In Section~\ref{sec:chromatic-number}, we additionally consider
networks of average degree \num{20}, using otherwise the same
parameter settings.

We note the power-law weights converge to uniform weights for
$\beta \to \infty$.  Thus, in the limit, the probability distribution
of the Chung--Lu model (almost)\footnote{The distributions are not
  exactly the same, as the Chung--Lu model achieves the desired number
  of edges only in expectation.  But conditioning on the average
  degree having its expected value, the two coincide.} coincides with
that of the Erdős--Rényi model.  It is thus not surprising that the
Erdős--Rényi graphs and the Chung--Lu graphs with $\beta = \num{25}$
occupy almost the same spot in
Figure~\ref{fig:params-of-generated-networks}.

\subparagraph{Geometric Ground Space.}

We usually use a torus with dimension $d = 2$ as ground space of the
GIRG model.  For temperature $T = 0$, we have a threshold model where
two vertices are connected if and only if their distance is
sufficiently small compared to the product of their weights.  Thus,
for $\beta \to \infty$ this converges to the commonly known model of
random geometric graphs; just on the torus $\mathbb T$ instead the
more commonly used unit square in the Euclidean plane.  Moreover, for
smaller $\beta$, we basically get hyperbolic random
graphs~\cite{kpk-hgcn-10}; just in one more dimension than usual.

The reason for using a torus instead of a square in the GIRG model is
that a square leads to special situations close to the boundary while
the torus wraps around and thus is completely symmetric.  This usually
simplifies theoretic analysis without changing too much otherwise.
Interestingly, we observe in our experiments that the choice of torus
vs.\ square makes a substantial difference for computing the diameter;
see Section~\ref{sec:diameter}.  To make this comparison, we
additionally generated five GIRGs with a square as ground space for
each of the parameter settings; see Section~\ref{sec:girgs-square} for
details on how we generated these networks.

\section{Comparison Between the Models and Real-World Networks}
\label{sec:comp-betw-real-world-and-models}

Each of the following subsections compares the performance of a
different algorithm between generated and real-world networks.  For
the cost $c$ of the algorithm, we plot $c$ depending on heterogeneity and locality
using color to indicate the cost; see, e.g., Figure~\ref{fig:bbbfs}.  The left and middle plot
show one data point for each parameter setting of the models and each
real-world network, respectively.  The right plot aggregates
real-world networks with similar locality and heterogeneity.  Each
point represents a number of networks indicated by its radius (log
scale).  We regularly assume the cost $c$ to be polynomial in $m$ (or
$n$), i.e., $c = m^x$, and plot $x = \log_m c$.  In this case, the
color in the left and right plot shows the mean exponent $x$
aggregated over the networks.

The set of real-world networks contains some networks with extreme heterogeneity; their inclusion in the plots would squeeze the non-extreme data points together and make interpretation of the middle and right plot more difficult. Thus, in these plots we only consider real-world networks with a heterogeneity value between~\num{-1.0} and~\num{1.5}. We refer to Appendix~\ref{sec:extreme-heterogeneity-ext} for full versions of such plots including all networks with extreme heterogeneity. This restriction only affects such plots; in our statistics and discussions, we consider all real-world networks, including those with extreme heterogeneity.

\subsection{Bidirectional Search}
\label{sec:bidirectional-search}

We can compute a shortest path between two vertices $s$ and $t$ using
a \emph{breadth first search (BFS)}.  The BFS explores the graph
layer-wise, where the $i$th layer contains the vertices of distance
$i$ from~$s$.  By \emph{exploring} the $i$th layer, we refer to the
process of iterating over all edges with an endpoint in the $i$th
layer, thereby finding layer $i+1$.  The search can stop once the
current layer contains $t$.

The \emph{bidirectional BFS} alternates the exploration of layers
between a \emph{forward BFS} from $s$ and a \emph{backward BFS}
from~$t$.  The alternation strategy we study here greedily applies the
cheaper of the two explorations in every step.  The cost of exploring
a layer is estimated via the sum of degrees of vertices in that layer.
The search stops once the current layers of forward and backward
search intersect.  The resulting algorithm is called \emph{balanced
  bidirectional BFS}~\cite{KADAB_ADapt_Algor_Betwe_jour2019}, which we
often abbreviate with just \emph{bidirectional BFS} in the following.

The \emph{cost} $c$ for the bidirectional BFS is the average number
of edge explorations over \num{100} random $st$-pairs.  Note that
$c \le 2m$, as each edge can be explored at most twice; once from each
side.  Figure~\ref{fig:bbbfs} shows the exponent $x$ for $c = m^x$
depending on heterogeneity and locality.

\begin{figure}
  \centering
  \includegraphics{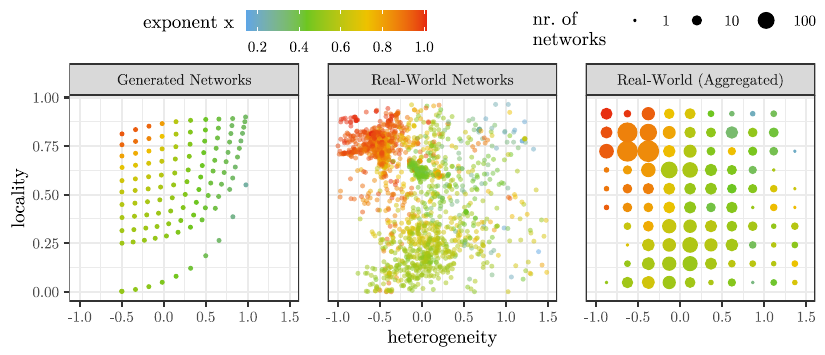}
  \caption{The exponent $x$ of the average cost $c = m^x$ of the
    bidirectional BFS over \num{100} $st$-pairs.}
  \label{fig:bbbfs}
\end{figure}

\subparagraph{Impact of Locality and Heterogeneity.}

For the generated data, we see that networks with high locality and
homogeneous degree distribution (top left corner) have an exponent of
around \num{1} (red).  Thus, the cost of the bidirectional BFS is
roughly $m$, which matches the worst-case bound.  If the network is
heterogeneous (right) or less local (bottom), we get significantly
lower exponents of around \num{0.5}, indicating a cost of roughly
$\sqrt{m}$.  The cost is particularly low for very heterogeneous
networks.  Overall, we get a strict separation between hard (high
locality, low heterogeneity) and easy (low locality or high
heterogeneity) instances.

For real-world networks, we observe the same overall behavior that
instances that are local and homogeneous tend to be hard while all
others tend to be easy.  There are only few exceptions to this,
indicating that the heterogeneity and locality are usually the crucial
properties impacting the performance of the bidirectional BFS.

\subparagraph{Discussion.}

Borassi and Natale~\cite{KADAB_ADapt_Algor_Betwe_jour2019} found that
the bidirectional BFS is surprisingly efficient on many networks,
which they used to efficiently compute the betweenness centrality.
Additionally, they studied the bidirectional BFS theoretically on
random network models where the edges are drawn independently, i.e.,
when there is no locality.  They in particular prove that the search
requires with high probability only $O(\sqrt{n})$ time for degree
distributions with bounded variance.  This includes the Erdős--Rényi
model, and the Chung--Lu or the configuration model with power-law
degree-distribution with power-law exponent $\beta \ge 3$.  For
$\beta \in (2, 3)$ (unbounded variance), the bound is $O(n^x)$ for
$x \in (0.5, 1)$\punctuationfootnote{\label{footnote:exponent}We note
  that the bounds for heterogeneous networks consider the worst-case
  over all $st$-pairs (which is dominated by the maximum degree for
  low power-law exponents) while our experiments consider the average
  over \num{100} random $st$-pairs.  This explains why the theoretic
  bounds, though sublinear, seem to be worse than the empirical
  bounds.}.

For graphs with locality, it is known that the bidirectional BFS
requires linear ($x = 1$) and sublinear
($x \in (0.5, 1)$)\footnotemark[\getrefnumber{footnote:exponent}] time
on geometric random graphs in Euclidean and hyperbolic space,
respectively~\cite{Effic_Short_Paths_Scale_jour2022}.  We note that
geometric random graphs in Euclidean and hyperbolic space essentially
correspond to the top-left and top-right corner of
Figure~\ref{fig:bbbfs}, as they can be viewed as special cases of
GIRGs without or with heterogeneity, respectively.

Overall, the theoretical results on network models cover all four
corners of the space spanned by heterogeneity and locality.  Moreover,
the predictions on the models match the observations on real-world
networks, i.e., for most networks, the practical run time of the
bidirectional BFS is not surprising but as expected.  This indicates
that locality and heterogeneity are the core deciding features for the
performance of the bidirectional BFS.

\subsection{Diameter}
\label{sec:diameter}

The \emph{eccentricity} of $s \in V$ is $\max_{t\in V}\dist(s, t)$,
i.e., the distance to the vertex farthest from $s$.  The
\emph{diameter} of the graph $G$ is the maximum eccentricity of its
vertices.  It can be computed using the \emph{iFUB
  algorithm}~\cite{compu_diame_realw_undir_jour2013}.  It starts with
a root $r \in V$ from which it computes the BFS tree $T$.  It then
processes the vertices bottom up in $T$ and computes their
eccentricities using a BFS per vertex.  This process can be pruned
when the distance to $r$ is sufficiently small compared to the largest
eccentricity found so far.  Pruning works well if $r$ is a central
vertex in the sense that it has low distance to many vertices and a
shortest path between distant vertices gets close to $r$.

There are different strategies of choosing the central vertex $r$.
The \emph{iFUB+hd} algorithm simply chooses a vertex of highest degree
as starting vertex.  A more sophisticated way of selecting $r$ works
as follows.  A \emph{double
  sweep}~\cite{Fast_compu_empir_tight_jour2008} starts with a vertex
$u$, chooses a vertex $v$ at maximum distance from $u$, and returns a
vertex $w$ from a middle layer of the BFS tree from~$v$.  A
\emph{4-sweep}~\cite{compu_diame_realw_undir_jour2013} consists of two
double sweeps, starting the second sweep with the result $w$ of the
first sweep.  The \emph{iFUB+4-sweephd} algorithm chooses $r$ by doing
a 4-sweep from a vertex of maximum degree.

The \emph{cost} $c$ of iFUB+hd and iFUB+4-sweephd is the number of
BFSs it performs (including the four initial BFSs of the 4-sweep for
iFUB+4-sweephd).  Note that $c \le n$.  Figure~\ref{fig:diameter-ifub}
shows the exponent $x$ for $c = n^x$ depending on heterogeneity and
locality.  For iFUB-hd and iFUB+4-sweephd, \num{17} and \num{16}
real-world networks, respectively, are excluded from the plots as they
exceeded the time limit of \SI{30}{min}.  The GIRG model uses a square
as ground space instead of the usual torus; see discussion below for
details.

\begin{figure}[p]
  \centering
  \begin{subfigure}{\textwidth}
    \centering
    \includegraphics{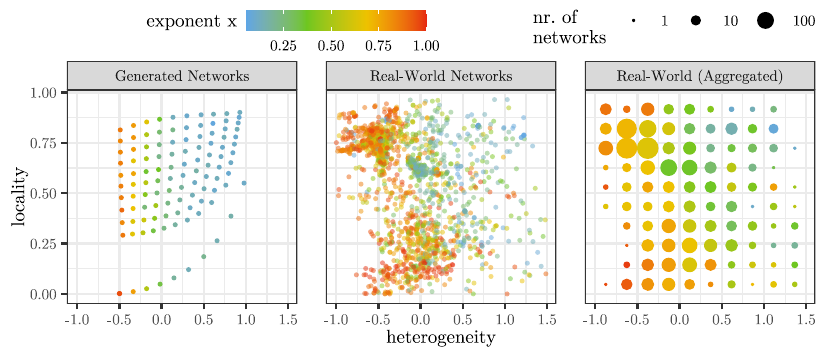}
    \caption{Results for the iFUB+hd algorithm; \num{17} real-world
      networks are excluded due to timeout.}
    \label{fig:diameter-ifub-hd}
  \end{subfigure}
  \begin{subfigure}{\textwidth}
    \centering
    \includegraphics{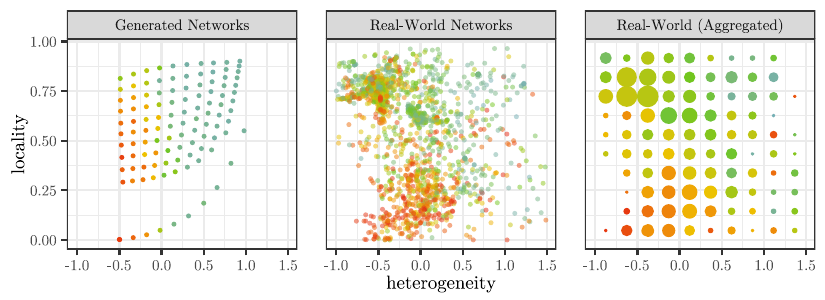}
    \caption{Results for the iFUB+4-sweephd algorithm; \num{16}
      real-world networks are excluded due to timeout.}
    \label{fig:diameter-ifub-foursweep}
  \end{subfigure}
  \begin{subfigure}{\textwidth}
    \hspace{1pt}
    \includegraphics{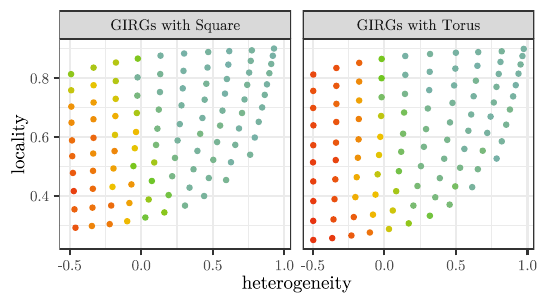}
    \caption{Comparison between torus and square as GIRG ground space
      for the iFUB+4-sweephd algorithm.}
    \label{fig:diameter-ifub-foursweep-torus}
  \end{subfigure}
  \caption{The exponent $x$ of the number of BFSs $c = n^x$ of the
    iFUB algorithms.  Different from the rest of the paper, the GIRG
    ground space is a square instead of a torus.}
  \label{fig:diameter-ifub}
\end{figure}

\subparagraph{Impact of Locality and Heterogeneity.}

The general dependence is the same for generated and real-world
networks.  For iFUB+hd (Figure~\ref{fig:diameter-ifub-hd}), an almost
linear number of BFSs is required for networks lacking heterogeneity.
On heterogeneous networks, i.e., if there are vertices of high degree,
the number of BFSs is substantially sublinear.  The more sophisticated
iFUB+4-sweephd variant (Figure~\ref{fig:diameter-ifub-foursweep})
additionally performs well on homogeneous networks with locality.
Observe that the picture for real-world networks shows some noise,
which indicates that there are properties besides locality and
heterogeneity that impact the performance.  We note, however, that
this observation agrees with the models, where we get highly varying
exponents for individual parameter settings, e.g., for iFUB+4-sweephd
on the five GIRGs with power-law exponent $\beta = 3.35$ and
temperature $T = 0.62$, we get exponents ranging from \num{0.18} to
\num{0.77} for the five generated instances.

\subparagraph{Impact of the Ground Space: Torus vs.\ Square.}

As mentioned in Section~\ref{sec:random-networks}, the default ground
space for the GIRG model is a torus.  Though this might seem less
natural than using a unit square, using the torus often simplifies
theoretical analysis while not making a big difference otherwise.
However, in the particular case of the iFUB+4-sweephd algorithm, the
ground space makes a significant difference.
Figure~\ref{fig:diameter-ifub-foursweep-torus} compares the efficiency
of iFUB+4-sweephd between torus-GIRGs and square-GIRGs.

Observe that the difference mainly shows for homogeneous and local
networks (top-left corner), where iFUB+4-sweephd requires an almost
linear number of BFSs on torus-GIRGs, while being sublinear for
square-GIRGs (as we have already observed in
Figure~\ref{fig:diameter-ifub-foursweep}).  This difference can be
explained as follows.  For square-GIRGs, the 4-sweep can find a vertex
in the center of the ground space, which has relatively low distance
to many vertices.  As mentioned above, starting iFUB from such a
vertex lets us prune the search early.  In a torus, however, all
points are identical and thus homogeneous torus-GIRGs do not have
central vertices, rendering the iFUB approach ineffective.

As mentioned above, iFUB+4-sweephd performs well on most homogeneous
and local real-world networks
(Figure~\ref{fig:diameter-ifub-foursweep}).  This indicates that
square-GIRGs are a better representation for these networks than
torus-GIRGs.  There are, however, exceptions.  Three such cases are
the local and homogeneous networks \texttt{man\_5976}, \texttt{tube1},
and \texttt{barth4} with exponents of \num{0.894}, \num{0.873}, and
\num{0.89} respectively.  Interestingly, all three graphs exhibit
torus-like structures in the sense that they ``wrap around'' in some
way, which obstructs the existence of a central vertex; see
Figure~\ref{fig:graph-drawings-diameter}.  Thus, also in this case,
the predictions of the models match the real-world behavior.

\begin{figure}[t!]
  \centering
  \includegraphics[width=0.3\linewidth]{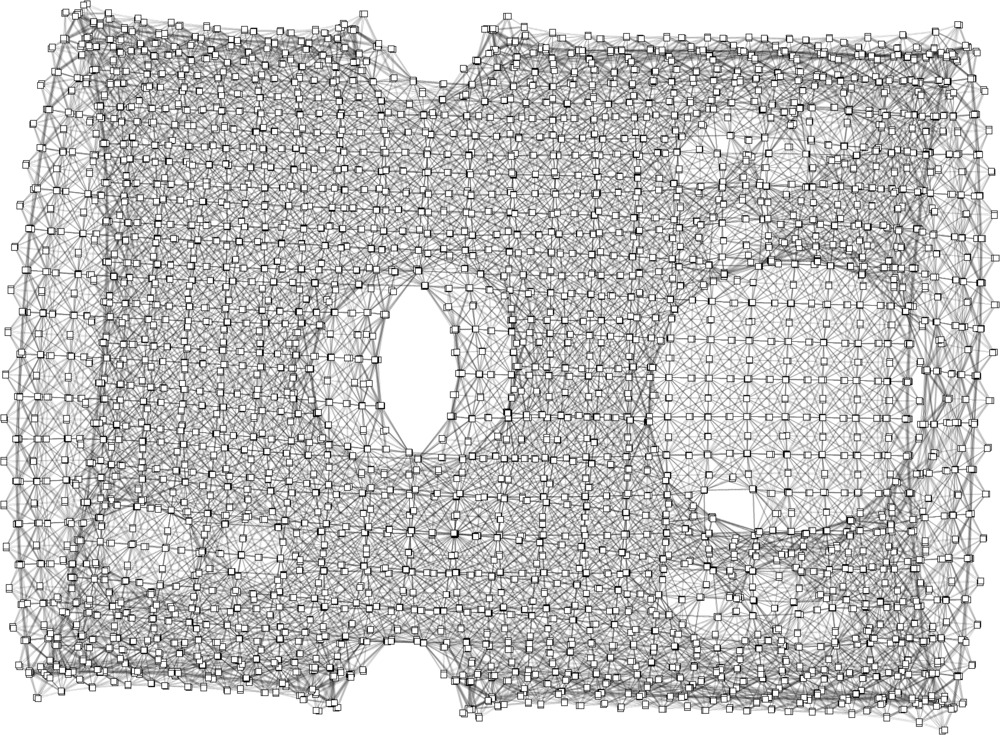}%
  \hfill%
  \includegraphics[width=0.3\linewidth]{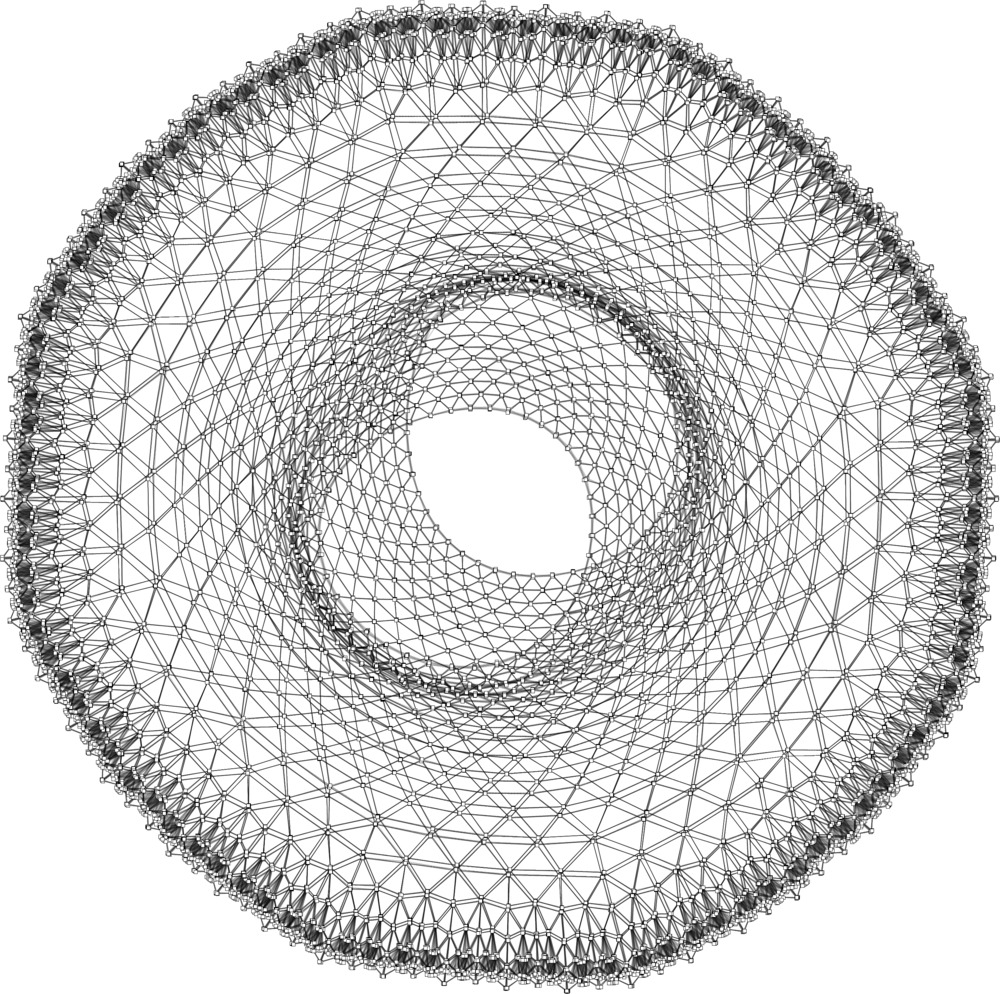}%
  \hfill%
  \includegraphics[width=0.3\linewidth]{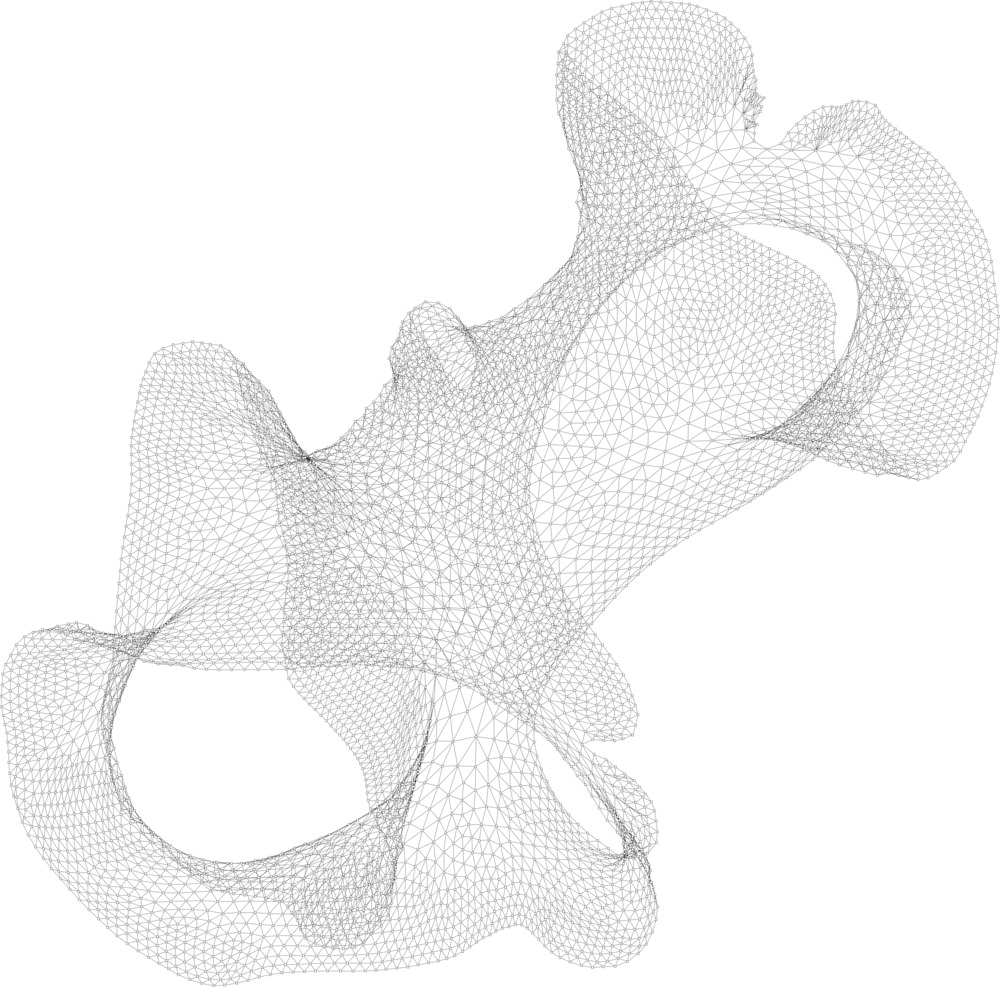}%
  \caption{Graph drawings of the real-world networks
    \texttt{man\_5976}, \texttt{tube1}, and \texttt{barth4} (from left
    to right).  The graphs were drawn with the OGDF
    implementation~\cite{Open_Graph_Drawi_Frame_other2013} of the
    $\text{FM}^3$ algorithm~\cite{Drawi_Large_Graph_with_GD2004}.}
  \label{fig:graph-drawings-diameter}
\end{figure}

\subparagraph{4-Sweep Lower Bound.}

Diameter approximation methods such as the double sweep lower bound
are known to perform well on real-world
graphs~\cite{Findi_Diame_RealW_Graph_ESA2010,
  Fast_compu_empir_tight_jour2008} and heterogeneous graph
models~\cite{Axiom_Avera_Analy_Algor_SODA2017}.  As our experiments
involve doing a 4-sweep anyways, we report the quality of the
resulting 4-sweep lower bounds on our data set of real-world networks.
We know the exact diameter for \num{2726} real-world networks (no
timeout in iFUB+hd or iFUB+4-sweephd).  The 4-sweep lower bound
matched the diameter for \num{2218} (\SI{81}{\%}) of these networks.
For \num{2703} (\SI{99}{\%}) networks, the difference between the
lower bound and the exact diameter is at most \num{2}.

\subparagraph{Discussion.}

The problem of computing the diameter of a graph is closely related to
the all pairs shortest path (APSP) problem, i.e., computing the
distance between all pairs of points.  It is an open problem whether
one can compute the diameter of arbitrary graphs faster than
APSP~\cite{c-dogopanr-87}.  The best known algorithms for APSP run in
time $O(n^\omega)$ with
$\omega<\num{2.38}$~\cite{Refin_Laser_Metho_Faste_SODA2021} or
$O(nm)$.  Both running times are infeasible for large real-world
networks.

In practice, however, Crescenzi, Grossi, Habib, Lanzi, and
Marino~\cite{compu_diame_realw_undir_jour2013} introduced the iFUB
algorithm and demonstrate that it performs much better on most
real-world networks than the worst case suggests.  They note that it
works particularly well on networks with high difference between
radius and diameter.  This corresponds to the existence of a central
vertex with eccentricity close to half the diameter.  Our experiments
confirm this and provide the following more detailed picture.  First,
in heterogeneous networks, the high degree nodes serve as central
vertices.  The results of Borassi, Crescenzi, and
Trevisan~\cite{Axiom_Avera_Analy_Algor_SODA2017} provide a theoretical
foundation for this phenomenon.  They in particular show that
heterogeneous Chung--Lu graphs allow the computation of the diameter
in sub-quadratic time, indeed using a vertex of high degree as central
vertex.  Second, for homogeneous networks, locality facilitates the
existence of a central vertex, unless the graph exhibits a toroidal
structure.

Our results point to some practical as well as theoretical open
questions.  The highest potential for practical improvements can be
achieved by focusing on homogeneous graphs without locality or with a
toroidal underlying geometry.  Erdős--Rényi graphs or random geometric
graphs with a torus as underlying geometry may help guide the
development of such an algorithm.  Moreover, it would be interesting
to strengthen the theoretical foundation by complementing the above
results by Borassi et al.~\cite{Axiom_Avera_Analy_Algor_SODA2017} on
Chung--Lu graphs to network models exhibiting locality (homogeneous as
well as heterogeneous).

\subsection{Vertex Cover Domination}
\label{sec:vert-cover-domin}

A vertex set $S \subseteq V$ is a \emph{vertex cover} if every edge
has an endpoint in $S$, i.e., removing $S$ from $G$ leaves a set of
isolated vertices.  We are interested in finding a vertex cover of
minimum size.  For two adjacent vertices $u, v \in V$, we say that $u$
\emph{dominates} $v$ if $N[v] \subseteq N[u]$.  The \emph{dominance
  rule} states that there exists a minimum vertex cover that includes
$u$.  Thus, one can reduce the instance by including $u$ in the vertex
cover and removing it from the graph.

To evaluate the effectiveness of the dominance rule, we apply it
exhaustively, i.e., until no dominant vertices are left.  Moreover, we
remove isolated vertices.  We refer to the number $c$ of vertices in
the largest connected component of the remaining instances as the
\emph{kernel size}.  Figure~\ref{fig:vertex-cover-domination} shows
the \emph{relative kernel size} $c / n$ with respect to locality and
heterogeneity.

\begin{figure}
  \centering
  \includegraphics{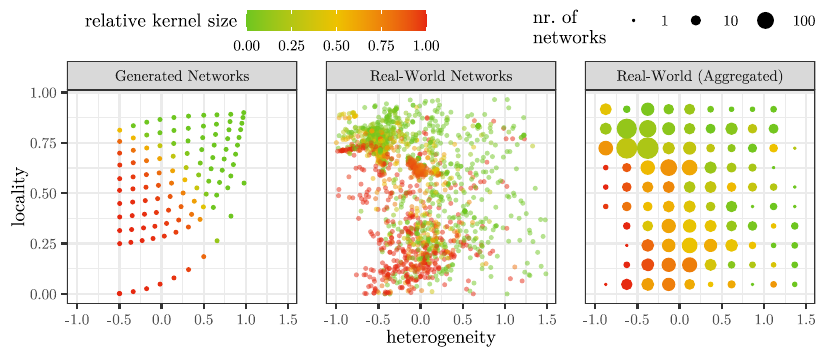}
  \caption{The relative kernel size of the vertex cover domination
    rule.}
  \label{fig:vertex-cover-domination}
\end{figure}

\subparagraph{Impact of Locality and Heterogeneity.}

We see a sharp separation for the generated networks.  For low
locality and heterogeneity (bottom left), the reduction rule cannot be
applied.  For high locality and heterogeneity (top right), the
dominance rule completely solves the instance.
For real-world networks, the separation is less sharp, i.e., there is
a larger range of locality/heterogeneity values in the middle where
the dominance rule is effective sometimes.  Nonetheless, we see the
same trend that the reduction rule is more likely to be effective the
higher the locality and heterogeneity.  In the extreme regimes (bottom
left or top right), we observe the same behavior as for the generated
networks with relative kernel sizes close to $1$ and $0$,
respectively, for almost all networks.  Moreover, there is dichotomy
in the sense that many instances are either (almost) completely solved
by the dominance rule or the rule is basically inapplicable.
In fact, \SI{30.9}{\%} of the real-world networks are reduced to below
\SI{5}{\%} of their original size, while \SI{16.3}{\%} are reduced
by less than \SI{5}{\%}.

\subparagraph{Discussion.}

Though vertex cover is NP-hard~\cite{Reduc_Among_Combi_Probl_CCC1972},
it is rather approachable: It can be solved in
$\num{1.1996}^nn^{O(1)}$ time~\cite{Exact_algor_maxim_indep_jour2017}
and there is a multitude of FPT-algorithms with respect to the
solution size $k$~\cite{Param_Algor_other2015}, the fastest one
running in
$O(\num{1.2738}^k + kn)$~\cite{Impro_upper_bound_verte_jour2010}.
This basis was improved by Harris and Narayanaswamy in a recent
preprint~\cite{hn-favcpss-22} stating a running time of
$O^*(\num{1.25400}^k)$ where $O^*$ suppresses polynomial factors.
Moreover, there are good practical
algorithms~\cite{Branc_expon_algor_pract_jour2016,w-ctspdr-98}.  Both
algorithms apply a suite of reduction rules, including the dominance
rule or a generalization.  We note that the dominance rule is closely
related to Weihe's reduction rules for hitting set~\cite{w-ctspdr-98}.
Our previous experiments for Weihe's reduction rules match our
results: they work well if the instances are local and
heterogeneous~\cite{Under_Effec_Data_Reduc_WAW2019}.
Concerning theoretic analysis on models, we know that on hyperbolic
random graphs, the dominance rule is sufficiently effective to yield a
polynomial time algorithm~\cite{bffk-svcpthrg-21}.
Thus, it is not surprising that the top-right corner in
Figure~\ref{fig:vertex-cover-domination} is mostly green.

We see two main directions for future research.  First, concerning
the dominance rule, we have seen that heterogeneity and locality are
a good predictor for the effectiveness in the extreme regimes.
However, there is a regime of moderate heterogeneity and locality
where we see a mix of high and low effectiveness.  This indicates
that other properties are the deciding factor for these instances
and it would be interesting to figure out these properties.  The
second direction is to investigate the effectiveness of other
techniques for solving vertex cover depending on the network
properties.  This includes the study of techniques used in existing
solvers~\cite{Branc_expon_algor_pract_jour2016,w-ctspdr-98} but also
the development of new techniques that are tailored towards less
local or less heterogeneous instances.  Specifically the algorithm
of Akiba and Iwata~\cite{Branc_expon_algor_pract_jour2016} performs
well on social networks (high locality and heterogeneity) but fails
on road networks (low heterogeneity).  Thus, techniques tailored
towards solving homogeneous instances could lead to an algorithm
that is more efficient on a wider range of instances.

\subsection{The Louvain Algorithm for Graph Clustering}
\label{sec:graph-clust-louv}

Let $V_1 \cupdot \dots \cupdot V_k = V$ be a \emph{clustering} where
each vertex set $V_i$ is a \emph{cluster}.  One is usually interested
in finding clusterings with dense clusters and few edges between
clusters, which is formalized using some quality measure, e.g., the
so-called modularity.  A common subroutine in clustering algorithms is
to apply the following local search.  Start with every vertex in its
own cluster.  Then, check for each vertex $v \in V$ whether moving $v$
into a neighboring cluster improves the clustering.  If so, $v$ is
moved into the cluster yielding the biggest improvement.  This is
iterated until no improvement can be achieved.  Doing this with the
modularity as quality measure (and subsequently repeating it after
contracting clusters) yields the well-known Louvain
algorithm~\cite{bgll-fucln-08}.

The run time of the Louvain algorithm is dominated by the number of
iterations of the initial local search.  Figure~\ref{fig:louvain}
shows this number of iterations.  It excludes four real-world networks
with more than \SI{1}{k} iterations (\num{1403}, \num{1403},
\num{5003}, and \num{19983}).

\begin{figure}
  \centering
  \includegraphics{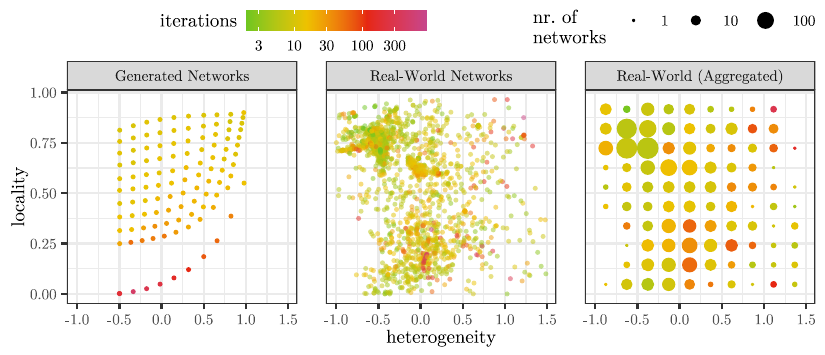}
  \caption{Number of iterations of the first local search of the
    Louvain algorithm.  The color scale is logarithmic.  Four outliers (real-world)
    with more than \SI{1}{k} iterations are excluded.}
  \label{fig:louvain}
\end{figure}

\subparagraph{Impact of Locality and Heterogeneity.}

For the generated instances, the number of iterations is generally low
but increasing when decreasing the locality.  To discuss this in more
detail, recall that each dot in the left plot represents the average
of five networks randomly generated with the same set of parameters.
Moreover, recall that the bottom-left dot represents Erdős--Rényi
graphs, the rest of the bottom row represents Chung--Lu graphs and all
other points represent GIRGs.

For GIRGs, the average number of iterations ranges from \num{9.4} to
\num{55.4} for the different parameter configurations.  Moreover, the
number of iterations starts to rise only for low localities.  For
\SI{83}{\%} of the configurations the average number of iterations
lies below \num{20}.  For the Erdős--Rényi graphs we obtain an average
of \num{196.4} iterations and for Chung--Lu graphs it goes up to
\num{443.40} for one configuration.  However, these average values
over five generated instances have to be taken with a grain of salt as
there is a rather high variance, e.g., the number of iterations for
the five Chung--Lu graphs with power-law exponent \num{25} ranges from
\num{104} up to \num{333}.

For the real-world networks there is no clear trend depending on
locality or heterogeneity.  In general, the number of iterations is
rather low except for some outliers.  While the strongest outlier
requires almost \SI{20}{k} iterations, \SI{98.6}{\%} of the networks
have at most \num{100} iterations.

\subparagraph{Discussion.}

The Louvain algorithm has been introduced by Blondel, Guillaume,
Lambiotte, and Lefebvre \cite{bgll-fucln-08}.  Discussing the plethora
of work building on the Louvain algorithm is beyond the scope of this
paper.  Concerning its running time, the worst-case number of
iterations of the Louvain algorithm can be upper bounded by $O(m^2)$
due to the fact that the modularity is bounded and each vertex move
improves it by at least $\Omega(1 / m^2)$.  Moreover, there exists a
graph family that requires $\Omega(n)$ iterations for the first local
search~\cite[Proposition~3.1]{Engin_Graph_Clust_Algor_other2015}.
Note that a linear number of iterations leads to a quadratic running
time, which is prohibitive for larger networks.

Our experiments indicate that locality or heterogeneity are not the
properties that discriminate between easy and hard instances.  For
generated instances, there is the trend that low locality increases
the number of iterations, which does not transfer to the real-world
networks (or is at least less clear).  However, the general picture
that most instances require few iterations while there are some
outliers coincides for generated and real-world networks.

We believe that the observed behavior can be explained as follows.
There are two factors that increase the number of iterations.  First,
one can craft very specific situations in which individual vertices
move into a cluster, trigger a change there and then continue into
another cluster.  These specific structures are unlikely to appear in
practice as well as in random instances, which is why we observe the
overall low number of iterations in our experiments.  The second
factor increasing the number of iterations is a lack of clear
community structure.  A vertex that fits more or less equally well to
different communities is likely to swap its cluster more often when
the clustering slightly changes during the local search.  This is why
we see an increased number of iterations on the generated networks
with low locality.  This interpretation is also supported by the fact
that Louvain performs provably well on the stochastic block
model~\cite{Power_Louva_Stoch_Block_NeurIPS2020}; a network model with
a very clear community structure of non-overlapping clusters.

For future research, we believe it is highly interesting to support
our above interpretation with theoretical evidence.  In particular, it
would be interesting to understand how the local search behaves on a
model with underlying geometry, which facilitates structures of
overlapping communities.

\subsection{Maximal Cliques}
\label{sec:maximal-cliques}

A \emph{clique} is a subset of vertices $C \subseteq V$ such that $C$
induces a complete graph, i.e., any pair of vertices in $C$ is
connected.  A clique is \emph{maximal}, if it is not contained in any
other clique.  In the following, we are never interested in
non-maximal cliques.  Thus, whenever we use the term \emph{clique}, it
implicitly refers to a maximal clique.  To enumerate all cliques, we
used the algorithm by Eppstein, Löffler, and
Strash~\cite{Listi_Maxim_Cliqu_Spars_ISAAC2010}, using their
implementation~\cite{Listi_Maxim_Cliqu_Large_jour2013,Listi_Maxim_Cliqu_Large_SEA2011}.

As the cliques of a network can be enumerated in polynomial time per
clique~\cite{Algor_Gener_Maxim_Indep_jour1977}, the number of maximal
cliques is a good indicator for the hardness of an instance.  To our
surprise, the number of cliques does not exceed the number of edges
$m$ for all generated and most real-world networks.  Out of the
\num{2740} real-world networks, \num{2552} (\SI{93}{\%}) have at most
$m$ and \num{187} have more than $m$ cliques.  For the remaining
\num{1} network, the timeout of \SI{30}{min} was exceeded.
Figure~\ref{fig:maximal-cliques} shows the number of cliques relative
to $m$ for all networks with at most $m$ cliques.

\begin{figure}
  \centering
  \includegraphics{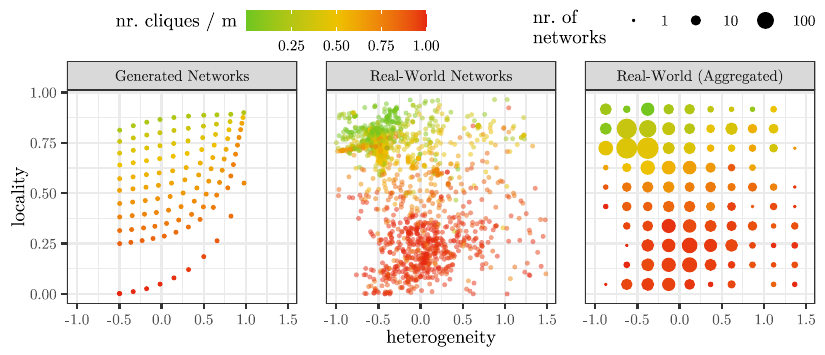}
  \caption{The number of maximal cliques relative to $m$ depending on
    the heterogeneity and locality, restricted to networks where this
    value is at most~\num{1} (\SI{93}{\%} of the networks).}
  \label{fig:maximal-cliques}
\end{figure}

\subparagraph{Impact of Locality and Heterogeneity.}

One can see that the networks (generated and real-world) with low
locality have roughly $m$ cliques, while the number of cliques
decreases for increasing locality.  Moreover, among networks with
locality, there is the slight trend that networks with higher
heterogeneity have more cliques.

It is not surprising that networks with low locality have roughly $m$
cliques, as graphs without triangles have exactly $m$ cliques.  For
graphs with higher locality, there are two effects counteracting each
other.  On the one hand, multiple edges combine into larger cliques,
which decreases the number of cliques.  On the other hand, each edge
can be contained in multiple cliques, which increases the number of
cliques.  Our experiments show that the former effect usually
supersedes the latter, i.e., the size of the cliques increases more
than the number of cliques each edge is contained in.

\subparagraph{Discussion.}

There are many results on enumerating cliques and on the complementary
problem of enumerating independent sets; discussing them all is beyond
the scope of this paper.  Here, we focus on discussing two results
that are closely related.  They are based on the degeneracy and the
so-called (weak) closure\footnote{For a formal definition of
  degeneracy and (weak) closure see the original papers or
  Section~\ref{sec:maximal-cliques-degen-closure}.} of a network.

The degeneracy is a robust measure for the sparsity of a network,
i.e., the degeneracy is low only if the graph does not include a dense
subgraph.  Eppstein, Löffler and
Strash~\cite{Listi_Maxim_Cliqu_Large_jour2013} give an algorithm for
enumerating all cliques that runs in $O(dn3^{d/3})$ time, where $d$ is
the degeneracy.  We use this algorithm for our experiments as it
performs very well in practice.

Fox, Roughgarden, Seshadhri, Wei, and
Wein~\cite{Findi_Cliqu_Socia_Netwo_jour2020} introduced the closure as
a measure that captures the tendency that common neighbors facilitate
direct connections, i.e., as a measure for locality.  Additionally,
they introduced the weak closure as a more robust measure.  Fox et
al.\ show that weakly $c$-closed graphs have at most
$3^{(c - 1)/3} n^2$ cliques.  For $c$-closed graphs they give the
additional upper bound of $4^{(c + 4)(c - 1) / 2} n^{2- 2^{1 - c}}$.
Together with an algorithm that takes polynomial time per
clique~\cite{Algor_Gener_Maxim_Indep_jour1977}, this shows that
enumerating all cliques is fixed-parameter tractable with respect to
the (weak) closure.

Comparing this to our empirical observations, the number of cliques is
at most $m$ for \SI{93}{\%} of real-world networks, while the
theoretical results give exponential bounds on how much bigger than
$m$ the number cliques can be.  Nonetheless, qualitatively speaking,
the bounds for degeneracy and closure show that sparsity and locality
lead to a low number of cliques.  This matches our observations on
real-world networks, i.e., most of them are sparse and have few
cliques and the number of cliques decreases with increasing locality.
Unfortunately, this interpretation does not withstand a closer look.
To study how indicative the degeneracy and (weak) closure actually are
for the number of cliques, we also computed\footnote{We give an
  algorithm for efficiently computing the weak closure in
  Section~\ref{sec:comp-weak-clos}.} these parameters for the networks
in our dataset.  We then study how the number of cliques depends on
these parameters as well as the relation between the parameters.  As
this is somewhat beyond the core focus of the paper, we only mention
the main insights here.  The detailed results can be found in
Appendix~\ref{sec:maximal-cliques-ext}.

For the \SI{93}{\%} of networks with at most $m$ cliques, the closure
does not correlate with the number of cliques.  Even worse, the weak
closure and the degeneracy have a slight negative correlation, i.e., a
lower parameter corresponds to a higher number of cliques.  On the
remaining \SI{7}{\%}-networks, the picture is somewhat different.
While the closure is still a bad indicator for the number of cliques,
the number of cliques is highly correlated with the degeneracy and the
weak closure.  Thus, degeneracy and weak closure provide a good
measure of hardness at least for the few somewhat hard instances with
more than $m$ cliques.  However, we note that degeneracy and weak
closure are almost identical on these networks.  Thus, weak closure is
mostly a measure of sparsity rather than of locality.

To summarize, although the existing theoretical results provide
valuable insights, they cannot explain the observed behavior on the
majority of the networks.  While providing strong upper bounds, these
parameterized results suffer from the same issue that a worst-case
analysis often has.  The worst case is usually harder than the typical
case and this does not change by restricting the set of considered
instances to those with small degeneracy or closure.  In contrast, our
experiments show that the network models yield surprisingly accurate
predictions.  They predict that one should generally expect sparse
networks to have at most $m$ cliques and that the number of cliques is
lower for networks with high locality.  Moreover, even the slight
increase with increasing heterogeneity is given for the generated as
well as for the real-world networks.  For future research we believe
that it is highly interesting to complement our empirical findings
theoretically by proving that the expected number of cliques is below
$m$.  Moreover, our findings reopen the question of finding
deterministic parameters that allow to give strong bounds on the
number of cliques while capturing most real-world networks.

\subsection{Chromatic Number}
\label{sec:chromatic-number}

The \emph{chromatic number} $\chi$ of $G$ is the smallest number of
colors one can use to color the vertices of $G$ without assigning
adjacent vertices the same color.  The size of any clique in $G$ is a
lower bound on the chromatic number.  Thus, for the \emph{clique
  number} $\omega$ of $G$, which is the size of the maximum clique in
$G$, we have $\chi \ge \omega$.  To compute $\chi$, it is safe to
reduce $G$ to its $\omega$-core, i.e., to iteratively remove all
vertices with degree less than $\omega$.  After coloring the resulting
\emph{kernel} optimally, one can always color the previously removed
vertices with at most $\omega$
colors~\cite{Solvi_Maxim_Cliqu_Verte_jour2015}.

The size of the kernel after applying this reduction is a good
indicator for the hardness of an instance.  For each network in our
data set, we first compute the clique number $\omega$ by enumerating
all maximal cliques as described in Section~\ref{sec:maximal-cliques}.
We then compute the $\omega$-core to obtain the kernel.  Let $c$ be
the number of vertices of the kernel.
Figure~\ref{fig:coloring-low-deg} shows the exponent $x$ for the
kernel size $c = n^x$, excluding \num{1} network where the clique
computation exceeded the timeout of \num{30}{min}.  As
we observed that the average degree of the generated instances has a
substantial impact on the kernel size, we additionally consider
generated instances with average degree \num{20} (instead of the usual
average degree of \num{10}).

\begin{figure}[t]
  \centering
  \begin{subfigure}{\textwidth}
    \centering
    \includegraphics{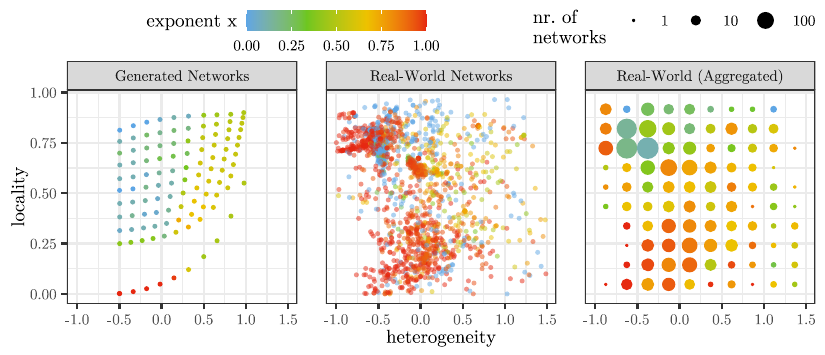}
    \caption{Dependence of the exponent $x$ on locality and
      heterogeneity.  The average degree for the generated networks is
      \num{10}.}
    \label{fig:coloring-low-deg}
  \end{subfigure}
  \begin{subfigure}{0.375\textwidth}
    \centering
    \includegraphics{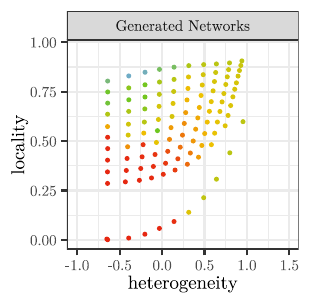}
    \caption{The same plot as in (\subref{fig:coloring-low-deg}) but
      with average degree \num{20} instead of \num{10}.}
    \label{fig:coloring-high-deg}
  \end{subfigure}
  \hfill
  \begin{subfigure}{0.55\textwidth}
    \centering
    \includegraphics{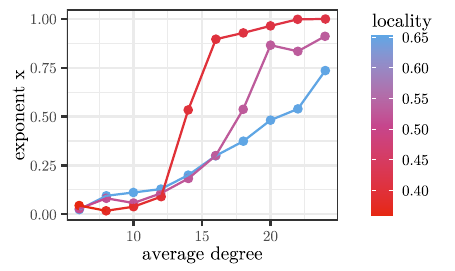}
    \caption{The exponent for GIRGs generated with $\beta = 5.1$ and
      $T\in \{0.62, 0.76, 0.88\}$ depending on the average degree.
      Each dot represents \num{30} networks.}
    \label{fig:chromatic-girg-scaled}
  \end{subfigure}
  \caption{The exponent $x$ of the kernel size $c = n^x$ after
    applying the chromatic number reduction rule.}
  \label{fig:chromatic-number}
\end{figure}

\subparagraph{Impact of Locality and Heterogeneity.}

For the generated networks, one can make three main observations.
First, graphs with higher locality have smaller kernels.  Second, for
highly heterogeneous networks, the kernel size is always moderate
independent of the locality.  Third, the average degree has a strong
impact on the applicability of the reduction rule, with much smaller
kernels for lower average degrees.

Many real-world instances are either heavily reduced or the reduction
rule is almost inapplicable, yielding a very noisy picture.  The
noisiness indicates that locality and heterogeneity are not the only
relevant factors for the kernel size.  However, the overall trend is
similar to the generated instances.  Moreover, the noisiness is also
present in the network models: The kernel size varies strongly with
the average degree but even for fixed average degree we observe high
variance, e.g., the exponents we get for the five generated GIRGs with
average degree \num{20}, power-law exponent \num{5.1} and temperature
\num{0.76} are \num{0.00}, \num{0.00}, \num{0.99}, \num{0.99}, and
\num{0.99}.

In the following, we first discuss the impact of the average degree in
more detail and then offer a potential explanation for the observed
variance in kernel size.

\subparagraph{Average Degree of Generated Instances.}

As mentioned above, the comparison of
Figures~\ref{fig:coloring-low-deg} and \ref{fig:coloring-high-deg}
shows that a higher average degree leads to larger kernels and thus
harder instances.  This dependence on the average degree can be seen
in more detail in Figure~\ref{fig:chromatic-girg-scaled}.  It shows
the exponent (averaged over 30 sampled instances) for GIRGs with three
parameter settings depending on the average degree.  For GIRGs with
high temperature (low locality) we see a threshold behavior going from
small kernels for low average degree to large kernels for high average
degree.  For lower temperatures (high locality), the transitions
happens later and less steeply.  This fits to the previously observed
separation between difficult non-local and easy local instances, which
moves up to higher locality values for an increasing average degree.

\subparagraph{Variance in the Kernel Size -- Clique Number and
  Degeneracy.}

Studying the clique number in comparison to the degeneracy yields a
potential explanation for the observed variance in the kernel size;
see Figure~\ref{fig:max-clique-size-and-degen}.

Recall that we use the clique number $\omega$ as lower bound for the
chromatic number.  Thus, a larger $\omega$ yields a better lower
bound, implying that the reduction rule can remove vertices of higher
degree.  However, the size of the kernel still depends on the
structure of the remaining graph, i.e., on whether there are
sufficiently many low-degree vertices that can be removed.  This
property is captured by the degeneracy $d$ of the graph.  In
particular, note that the degeneracy is lower-bounded by $\omega - 1$.
If the degeneracy is exactly $\omega - 1$, then the reduction rule
completely solves the instance.  If the degeneracy is higher, the
kernel can be very large.

\begin{figure}[t!]
  \centering
    \begin{subfigure}{\textwidth}
    \centering
    \includegraphics{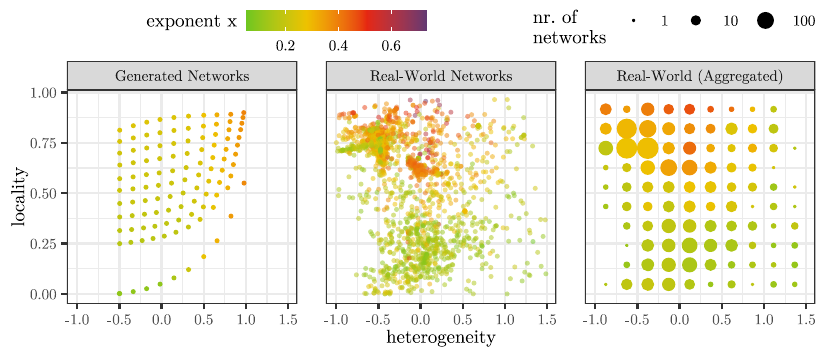}
    \caption{The exponent $x$ of the clique number $\omega = n^x$
      depending on locality and heterogeneity.}
    \label{fig:max-clique-size-and-degen-clique}
  \end{subfigure}
  \begin{subfigure}{\textwidth}
    \centering
    \includegraphics{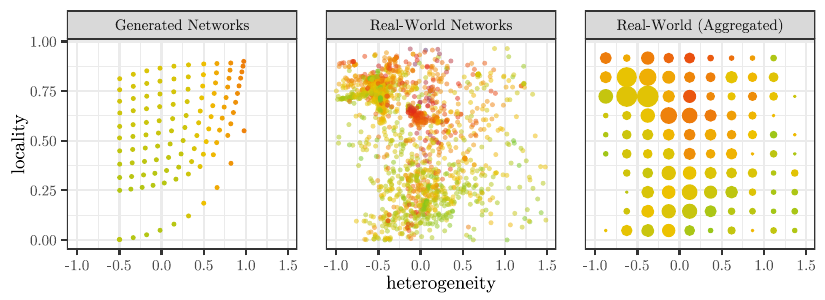}
    \caption{The exponent $x$ of the degeneracy $d = n^x$ depending on
      locality and heterogeneity.}
    \label{fig:max-clique-size-and-degen-degen}
  \end{subfigure}
  \caption{Comparison between the clique number and the degeneracy.}
  \label{fig:max-clique-size-and-degen}
\end{figure}

In Figure~\ref{fig:max-clique-size-and-degen}, we can make multiple
observations.  Locality and heterogeneity seem to be the core factors
impacting the clique size and the degeneracy, and the dependence is
similar for generated and real-world networks.  However, the
dependence is almost identical for clique size and degeneracy.  Thus,
whether $d = \omega - 1$, which yields an empty kernel or
$d > \omega - 1$, yielding a non-empty kernel, often comes down to a
property other than locality and heterogeneity, which explains the
high variance observed above.  Also observe how the difference between
clique size and degeneracy is slightly higher for less local graphs,
which fits to the previous observation that non-local graphs tend to
have larger kernels.

\subparagraph{Discussion.}

The reduction rule we considered here has been described by Verma,
Buchanan, and Butenko~\cite{Solvi_Maxim_Cliqu_Verte_jour2015}, who
demonstrated that the clique number is a good lower bound for the
chromatic number on real-world instances.  Another lower bound method
based on degree-bounded independent sets performing very well in
practice has been introduced by Lin, Cai, Luo, and
Su~\cite{Reduc_based_Metho_Color_IJCAI2017}. Concerning the clique number, Walteros and Buchanan~\cite{Maxim_Cliqu_Often_Easy_jour2020} have shown that in many real-world networks, the \emph{clique-core gap} $g = d+1-\omega$ is small, and provide an $O(m^{1.5})$ algorithm for the clique number if this gap is constant. However, they leave a potential parameterization for the chromatic number based on the clique-core gap as an open problem.

Concerning the impact of locality and heterogeneity, our results
reveal a highly interesting effect.  We have two characteristics
(clique number and degeneracy), where locality and heterogeneity are
the crucial factors.  However, for an algorithm whose efficiency is
essentially based on the comparison of the two characteristics,
properties besides locality and heterogeneity have to tip the scale
for whether the algorithm performs well or not.  This is an
interesting insight as this effect might pose a threat to external
validity when considering more complicated algorithms based on
multiple structural properties.  In this specific case, however, the
behavior observed on real-world instances is nonetheless quite similar
to that on generated instances, including the high variance due to the
above effect.

For future research, it would be interesting to study the reduction
rule based on maximum cliques theoretically on the different network
models.  We believe that this can give a better understanding of what
properties tip the scale in case the degeneracy is higher than the
clique number, potentially leading to improved algorithms.  This can
also lead to interesting technical contributions like, e.g.,
understanding the potential threshold behavior when it comes to the
dependence on the average degree.

\section{Discussion and Conclusion}
\label{sec:disc-concl}

Networks from different domains have varying structural properties.
Thus, trying to find probability distributions that match or closely
approximate those of real-world networks seems like a hopeless
endeavor.  Moreover, even if we knew such a probability distribution,
it would most likely be highly domain specific and too complicated for
theoretical analysis.

\subparagraph{Our View on Average-Case Analysis.}

A more suitable approach to average case analysis is the use of models
that assume few specific structural properties and are as unbiased as
possible beyond that.  If the chosen properties are the dominant
factors for the algorithm's performance, we obtain external validity,
i.e., the results translate to real-world instances even though they
do not actually follow the assumed probability distributions.  There
are two levels of external validity.
\begin{itemize}
\item The models capture the performance-relevant properties
  sufficiently well that algorithms perform similar on generated and
  real-world networks.
\item The models are too much of an idealization for this direct
  translation to practical performance.  However, varying a certain
  structural property in this idealized world has the same qualitative
  effect on performance as it has on real-world networks.
\end{itemize}
Though an average case analysis cannot yield strong performance
guarantees, with the above notions of external validity, it can give
insights into what properties are crucial for performance (first
level) and how the performance depends on a property (second level).
Moreover, even a lack of external validity can yield valuable insights
in the following sense.  Assume we have the hypothesis that property
$X$ is the crucial factor for algorithmic performance and thus we
study a model with property $X$ as null model.  Then, a lack of
external validity lets us refute the hypothesis as there clearly has
to be a different property impacting the performance.

\subparagraph{Impact of Locality and Heterogeneity.}

Non-surprisingly, the performance on real-world networks depending on
locality and heterogeneity is more noisy compared to the generated
networks, as real-world networks are diverse and vary in more than
just these two properties.  That being said, the observations on the
models and in practice coincide almost perfectly for the
bidirectional search and the enumeration of maximal cliques.  For the
maximal cliques, the match even includes constant factors, which is
particularly surprising, as these numbers are below $m$ while the
worst case is exponential.

For the vertex cover domination rule as well as the iFUB algorithm, we
obtain a slightly noisier picture.  However, the overall trend matches
well, which indicates that locality and heterogeneity are crucial
factors for the performance, but not the only ones.  For the iFUB
algorithm, we already identified the existence of central vertices as
additional factor (difference between torus and square as underlying
geometry).

For the chromatic number, we observe that heterogeneity and locality
are important but not the only factors impacting performance.
Nonetheless, the performance is similar on the models compared to
real-world networks.

For the Louvain clustering algorithm, the models and real-world
networks coincide insofar, that the number of iterations is low, with
few exceptions.  This indicates that locality and heterogeneity are
not the core properties for differentiating between hard and easy
instances.

\subparagraph{Conclusions.}

Locality and heterogeneity have significant impact on the performance
of many algorithms.  We believe that it is useful for the design of
efficient algorithms to have these two dimensions of instance
variability in mind.  Moreover, GIRGs~\cite{Geome_inhom_rando_graph_jour2019}
with the available efficient
generator~\cite{Effic_Gener_Geome_Inhom_ESA2019} provide an abundance
of benchmark instances on networks with varying locality and
heterogeneity.  Finally, we believe that average case analyses on the
four extreme cases can help to theoretically underpin practical run
times.  The four extreme cases can, e.g., be modeled using geometric
random graphs for local plus homogeneous, hyperbolic random
graphs\footnote{GIRGs can be seen as common generalization of
  geometric and hyperbolic random graphs.} for local plus
heterogeneous, Erdős--Rényi graphs for non-local plus homogeneous, and
Chung--Lu graphs for non-local plus heterogeneous networks.

\bibliography{locality-heterogeneity-auto,locality-heterogeneity-manual}

\newpage
\appendix

\section{Appendix -- Overview}

The appendix provides additional details and findings that divert from
the core focus of the paper but are nonetheless worth mentioning.
Appendix~\ref{sec:locality-ext} contains additional details on
the locality measure.  In Appendix~\ref{sec:average-distance}, we
discuss how we approximate the average distance of a graph.
Appendix~\ref{sec:random-networks-details} contains some technical
details on the randomly generated networks.
Appendix~\ref{sec:maximal-cliques-ext} studies the number of maximal
cliques depending on the degeneracy and (weak) closure parameters.
Appendix~\ref{sec:extreme-heterogeneity-ext} contains figures including the real-world networks with extreme heterogeneity.

The following list highlights some findings contained in the appendix
that we believe are interesting in their own right.
\begin{itemize}
\item We describe how we can compute the locality of a graph.  This in
  particular involves computing the detour distance quickly, which we
  can do exactly, using the knowledge about the bidirectional search
  from Section~\ref{sec:bidirectional-search}.
\item To compute the locality, we also have to compute the average
  distance of a graph.  For this, we provide an improved version
  of a previously known approximation
  algorithm~\cite{Avera_Dista_Queri_throu_APPROX2015} based on
  weighted sampling.  Moreover, we observe that uniform sampling might
  be more efficient at least if the bidirectional search is fast.
\item We provide an efficient algorithm for computing the weak closure
  of a graph.
\item We give a detailed comparison of the weak closure with other
  measures (in particular with the degeneracy).
\end{itemize}

\section{Locality}
\label{sec:locality-ext}

This section contains an extended discussion of our measure of
locality.  In Section~\ref{sec:comp-with-clust-details}, we compare
locality with the more commonly known clustering coefficient.  In
Section~\ref{sec:limitations-locality-details}, we discuss limitations
of the measure.  In Section~\ref{sec:comp-edge-local-details}, we
describe how we approximate the locality.

\subsection{Comparison With the Clustering Coefficient}
\label{sec:comp-with-clust-details}

Let $G = (V, E)$ be a graph and let $v \in V$ be a vertex with
$\deg(v) > 1$.  Let further $G[N(v)]$ be the graph induced by $v$'s
neighborhood $N(v)$.  The \emph{local clustering coefficient} of $v$
is the density of $G[N(v)]$, i.e., the number of edges of $G[N(v)]$
divided by $\deg(v) \choose 2$.  Thus, the local clustering
coefficient of $v$ ranges between \num{0} if its neighborhood is an
independent set and \num{1} if it is a clique.  The \emph{average
  local clustering coefficient} (or just \emph{clustering coefficient}
for short) of $G$ is the average over all vertices with degree at
least \num{2}\punctuationfootnote{Note that the local clustering
  coefficient is undefined for vertices of lower degree as we divide
  by ${\deg(v) \choose 2} = 0$ if $\deg(v) \le 1$.}.

Note that, roughly speaking, the clustering coefficient is high if the
graph contains many triangles.  A triangle $uvw$ contributes the edge
$\{v, w\}$ to the neighborhood of $u$, thereby increasing its local
clustering coefficients (and symmetrically for $v$ and $w$).  This
shows the similarity to the degree locality we introduced in
Section~\ref{sec:definition-locality}.  There, the triangle $uvw$ makes
$w$ a common neighbor of the edge $\{u, v\}$, increasing the degree
locality of $\{u, v\}$ (and symmetrically for $\{v, w\}$ and
$\{w, u\}$).  It is thus not surprising that the clustering
coefficient and the degree locality of graphs behave very similar as
can be seen in the left plot of
Figure~\ref{fig:locality_heterogeneity_clustering_app} (which is the
same as Figure~\ref{fig:locality_heterogeneity_clustering} but
reappears here for better readability).

\begin{figure}
  \centering
  \includegraphics{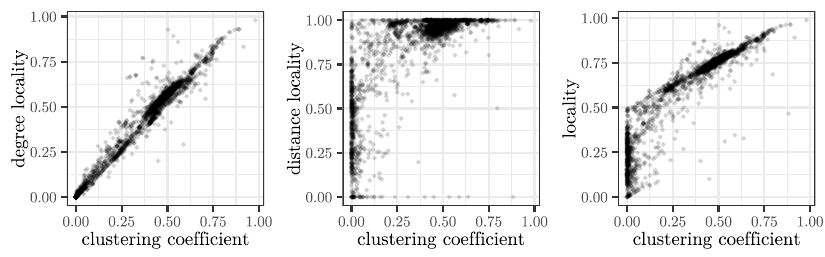}
  \caption{Comparison of the average local clustering coefficient to
    the degree locality (left), the distance locality (center), and
    the locality (right).  Each dot represents one network from our
    data set of real-world networks.}
  \label{fig:locality_heterogeneity_clustering_app}
\end{figure}

The main disadvantage of the clustering coefficient (and of the degree
locality for that matter) is that some types of graphs we would regard
as highly local get very small values.  An extreme case are bipartite
graphs that have no triangles and thus clustering coefficient and
degree locality \num{0}.  However, we would regard, e.g., grids as
highly local.  The goal of bringing the distance locality into the mix
was to get a more fine-grained view on these networks with few
triangles while keeping a measure similar to the clustering
coefficient for networks with many triangles.

In the middle plot of
Figure~\ref{fig:locality_heterogeneity_clustering_app}, we can see that
the networks with low clustering coefficients yield a wide range of
distance locality values.  Moreover, most networks with higher
clustering coefficient have distance locality close to~\num{1}.  This
should not come as a surprise as the distance locality of an edge is
\num{1} for every edge contained in a triangle.

Averaging the degree and distance locality yields the locality, which
is compared to the clustering coefficient in the right plot of
Figure~\ref{fig:locality_heterogeneity_clustering_app}.  This plot shows
that our measure of locality achieved what we hoped for.  For networks
with high clustering it behaves similar to the clustering coefficient,
but scaled to the interval $[0.5, 1]$.  Moreover, graphs with
clustering coefficient close to \num{0} are mapped to lower localities
in the $[0, 0.5]$ range, yielding a more find-grained perspective for
those networks.

Note that the above mentioned example of a grid has a locality close
to \num{0.5}.  If $G = (V, E)$ is a $\sqrt{n}\times\sqrt{n}$-grid,
then every edge $\{u, v\} \in E$ has detour distance \num{3} while the
average distance of non-adjacent vertex pairs grows with growing $n$.
This yields a distance locality that converges to \num{1} for
$n \to \infty$.

\subsection{Limitations}
\label{sec:limitations-locality-details}

Our definition of locality has three limitations.  We mentioned them
earlier in Section~\ref{sec:data-set} but want to discuss them here in more detail.

The first limitation is that the locality is not defined for trees.
In a tree, every edge is a bridge but the distance and edge locality
are only defined for non-bridge edges.  However, we do not see this as
a real issue as the concept of locality does not really make sense for
trees anyways.

The second limitation is that distance locality is not defined for
cliques.  In a clique, there are no non-edges and thus the looking at
the average non-edge distance does not make sense.

As already mentioned in Section~\ref{sec:data-set}, we
circumvent these first two limitations by removing trees and dense
graphs.  We note that trees are not very interesting in the context of
this paper and dense graphs are out of scope.

The third limitation is less clear-cut but also concerns the average
non-edge distance $\dist(\overline{E})$ in the definition of distance
locality.  As mentioned before, it can happen that
$\dist(\overline{E}) = 2$, e.g., for a wheel graph.  In this case we
formally defined $\distloc{\{u, v\}} = 0$.  However, there is also
already an issue if $\dist(\overline{E})$ is only slightly greater
than $2$.  In this case, the distance locality is rather volatile,
yielding negative values without much meaning.  This issue is
increased by the fact that we approximate the average distances.
Generally, the distance locality seems not the best measure for
locality in the regime where it is negative.

We circumvent this issue by capping the distance locality of graphs
at~$0$, as mentioned in Section~\ref{sec:definition-locality}.  We
note, however, that this happens only rarely.  Out of all networks,
the uncapped distance locality is negative for only \num{43} networks.
Moreover, the average non-edge distance is $2$ for \num{3} additional
networks.  In Table~\ref{tab:networks_with_neg_locality}, we list all
networks with average non-edge distance $2$ and all networks with
distance locality below \num{-0.1}.  Note that there are two types of
instances where this appears.  The instances at the top of the table
have low average distances, which makes the distance locality
non-robust.  The instances at the bottom are close to trees with few
long cycles, which yields large detour distances.

\begin{table}[t]
  \caption{The networks with (uncapped) distance locality below
    $\num{-0.1}$ or average non-edge distances~$2$.  The columns are
    the number of vertices $n$, number of edges $m$, the average
    distance $\dist(E \cup \overline{E})$, the average non-edge
    distance $\dist(\overline{E})$, the average detour distance
    $\dist^+(E')$ of non-bridges $E'$, the (uncapped) distance
    locality $\distloc{G}$ of $G$, and the degree locality
    $\degloc{G}$.}
  \label{tab:networks_with_neg_locality}
  \centering%
  \setlength{\tabcolsep}{0.75\tabcolsep}%
  \footnotesize%
  % latex table generated in R 4.2.0 by xtable 1.8-4 package
% Fri Jun 16 13:54:34 2023
\begin{tabular}{lrrrrrrr}
  \toprule
graph & $n$ & $m$ & $\dist(E\cup\overline{E})$ & $\dist(\overline{E})$ & $\dist^+(E')$ & $\distloc{G}$ & $\degloc{G}$ \\ 
  \midrule
Chebyshev1 & \num{261} & \num{1542} & \num{1.95} & \num{2.00} & \num{2.00} & \num{0.00} & \num{0.93} \\ 
  fs\_541\_1 & \num{541} & \num{2466} & \num{1.98} & \num{2.00} & \num{2.00} & \num{0.00} & \num{0.62} \\ 
  lp\_d6cube & \num{6184} & \num{37681} & \num{2.00} & \num{2.00} & \num{2.00} & \num{0.00} & \num{0.66} \\ 
  lp\_fit2d & \num{10524} & \num{129040} & \num{2.00} & \num{2.00} & \num{2.16} & \num{-116.88} & \num{0.15} \\ 
  bibd\_17\_8 & \num{24310} & \num{680232} & \num{2.00} & \num{2.00} & \num{2.01} & \num{-0.53} & \num{0.40} \\ 
  lp\_fit1d & \num{1049} & \num{13426} & \num{1.98} & \num{2.01} & \num{2.14} & \num{-17.96} & \num{0.13} \\ 
  arc130 & \num{130} & \num{715} & \num{1.98} & \num{2.07} & \num{2.15} & \num{-1.18} & \num{0.69} \\ 
  air02 & \num{6774} & \num{61529} & \num{2.19} & \num{2.19} & \num{2.35} & \num{-0.81} & \num{0.17} \\ 
  rt\_lolgop & \num{9765} & \num{10075} & \num{2.23} & \num{2.23} & \num{2.40} & \num{-0.78} & \num{0.54} \\ 
  nw14 & \num{123409} & \num{904906} & \num{2.31} & \num{2.31} & \num{2.94} & \num{-2.07} & \num{0.01} \\ 
  rt\_occupywallstnyc & \num{3609} & \num{3830} & \num{2.37} & \num{2.37} & \num{2.53} & \num{-0.42} & \num{0.41} \\ 
  blockqp1 & \num{60012} & \num{300011} & \num{2.44} & \num{2.44} & \num{2.67} & \num{-0.51} & \num{0.20} \\ 
  stat96v4 & \num{63076} & \num{491329} & \num{2.47} & \num{2.47} & \num{2.56} & \num{-0.18} & \num{0.13} \\ 
  bibd\_9\_3 & \num{84} & \num{249} & \num{2.56} & \num{2.68} & \num{2.86} & \num{-0.27} & \num{0.06} \\ 
  rt\_barackobama & \num{9631} & \num{9772} & \num{2.84} & \num{2.84} & \num{3.37} & \num{-0.64} & \num{0.15} \\ 
  n3c5-b4 & \num{252} & \num{1165} & \num{2.84} & \num{2.91} & \num{3.02} & \num{-0.12} & \num{0.01} \\ 
  rt\_onedirection & \num{7987} & \num{8100} & \num{2.91} & \num{2.91} & \num{3.25} & \num{-0.36} & \num{0.09} \\ 
  lp\_standmps & \num{1274} & \num{3878} & \num{2.98} & \num{2.99} & \num{3.10} & \num{-0.12} & \num{0.05} \\ 
  lp\_afiro & \num{51} & \num{100} & \num{2.88} & \num{3.04} & \num{3.16} & \num{-0.12} & \num{0.12} \\ 
  lp\_standata & \num{1274} & \num{3230} & \num{3.09} & \num{3.09} & \num{3.28} & \num{-0.17} & \num{0.00} \\ 
  primagaz & \num{10836} & \num{20116} & \num{3.27} & \num{3.27} & \num{3.77} & \num{-0.39} & \num{0.00} \\ 
  lp\_sc50a & \num{77} & \num{155} & \num{3.25} & \num{3.38} & \num{3.59} & \num{-0.16} & \num{0.01} \\ 
  lp\_sc50b & \num{76} & \num{143} & \num{3.38} & \num{3.51} & \num{3.78} & \num{-0.18} & \num{0.02} \\ 
  lpi\_ex72a & \num{215} & \num{463} & \num{3.94} & \num{4.00} & \num{4.24} & \num{-0.12} & \num{0.01} \\ 
  lpi\_woodinfe & \num{89} & \num{138} & \num{4.12} & \num{4.24} & \num{4.76} & \num{-0.24} & \num{0.11} \\ 
  ENZYMES123 & \num{90} & \num{127} & \num{6.13} & \num{6.29} & \num{7.00} & \num{-0.16} & \num{0.15} \\ 
  ENZYMES & \num{125} & \num{141} & \num{12.94} & \num{13.16} & \num{17.49} & \num{-0.39} & \num{0.01} \\ 
  NCI1 & \num{106} & \num{107} & \num{13.04} & \num{13.28} & \num{17.57} & \num{-0.38} & \num{0.00} \\ 
  FRANKENSTEIN & \num{214} & \num{217} & \num{22.25} & \num{22.46} & \num{40.08} & \num{-0.86} & \num{0.00} \\ 
  pivtol & \num{102} & \num{103} & \num{25.50} & \num{26.01} & \num{97.15} & \num{-2.96} & \num{0.02} \\ 
  odepa400 & \num{400} & \num{402} & \num{99.75} & \num{100.25} & \num{391.10} & \num{-2.96} & \num{0.01} \\ 
   \bottomrule
\end{tabular}

\end{table}

\subsection{Computing the Locality}
\label{sec:comp-edge-local-details}

In our implementation, we compute three aggregated values; the average
degree locality $\degloc{G}$ of the graph $G$, the average detour
distance $\dist^+(E')$ of all non-bridges $E'$, and the average
distance $\dist(E \cup \overline{E})$ of all vertex pairs (recall that
$\overline E = {V \choose 2} \setminus E$).  In the following, we
first discuss how we compute these values and then prove some lemmas
that show how these three values suffice to compute the locality
$\loc{G}$ of $G$.

Computing the degree locality $\degloc{G}$ in time
$O(\sum_{v \in V}(\deg(v))^2)$ is more or less straight forward (very
similar to computing the clustering coefficient).  This yields run
times that are feasible for all networks in our data set.

For the average detour distance $\dist^+(E')$ of all non-bridges $E'$,
we generally have to compute a shortest path for a linear number of
start--destination pairs.  Though this appears to require quadratic
running time, which would be infeasible, we can make use of the
following win--win situation.  We know that the bidirectional search
is sublinear unless the network is homogeneous and local; see
Section~\ref{sec:bidirectional-search}.  Moreover, if the network is
homogeneous and local, the detour distances are short and thus the
shortest path search can terminate after few steps.  Thus, just
running the bidirectional search for each non-bridge edge $E'$ is
efficient for all networks.  

The straight-forward way of computing the average distance
$\dist(E \cup \overline{E})$ of all vertex pairs is to run a BFS from
every vertex.  As this is infeasible for the larger networks, we
instead only approximate the average distance.  Details on that can be
found in Section~\ref{sec:average-distance}.

The following lemma states how we can compute the distance locality of
the graph from the average detour distance $\dist^+(E')$ of all
non-bridges $E'$ and the average distance $\dist(\overline{E})$ of
non-edges $\overline{E}$.

\begin{lemma}
  The average distance locality $\distloc{G}$ is
  \begin{equation*}
    \distloc{G} = \max\left\{ 1 - \frac{\dist^+(E') -
        2}{{\dist(\overline{E}) - 2}}, 0 \right\}.
  \end{equation*}
\end{lemma}
\begin{proof}
  First recall from Section~\ref{sec:definition-locality} that the
  maximum with $0$ is part of the definition of $\distloc{G}$.  Beyond
  that, we have to show that computing the distance locality of a
  single edge based on its detour distance commutes with taking the
  average over all non-bridge edges.  We obtain
  \begin{align*}
    \distloc{G}
    &= \frac{1}{m'} \sum_{\{u, v\} \in E'} \distloc{\{u, v\}}\\
    &= \frac{1}{m'} \sum_{\{u, v\} \in E'}
      \left(1 - \frac{\dist^+(u, v) - 2} {\dist(\overline E) -
      2}\right) \\ 
    &= 1 - \frac{1}{m'} \sum_{\{u, v\} \in E'}\frac{\dist^+(u, v) - 2}{\dist(\overline{E}) - 2}\\
    &= 1 - \frac{\frac{1}{m'} \sum_{\{u, v\} \in E'} \dist^+(u, v) - 2}{\dist(\overline{E}) - 2}\\
    &= 1 - \frac{\dist^+(E') - 2}{\dist(\overline{E}) - 2}.
  \end{align*}
\end{proof}

Finally, the following lemma shows how we can derive the average
distance $\dist(\overline{E})$ of non-edges $\overline{E}$ with
$\overline{m} = |\overline{E}|$ from the average distance
$\dist(E \cup \overline E)$ of all vertex pairs.

\begin{lemma}
  The average distance $\dist(\overline{E})$ between non-edges is
  \begin{equation*}
    \dist(\overline{E}) = \dist(E \cup \overline E) +
    \frac{m}{\overline{m}} (\dist(E \cup \overline E) - 1).
  \end{equation*}
\end{lemma}
\begin{proof}
  It holds that
  \begin{align*}
    \dist(\overline{E})
    &= \frac{1}{\overline{m}} \sum_{\{u, v\} \in \overline{E}} \dist(u, v)\\
    &= \frac{1}{\overline{m}} \left(
      \sum_{\{u, v\} \in \overline{E}} \dist(u, v)
      + \sum_{\{u, v\} \in E} \dist(u, v)
      - \sum_{\{u, v\} \in E} \dist(u, v) \right)\\
    &= \frac{1}{\overline{m}} \sum_{\{u, v\} \in E \cup \overline{E}} \dist(u, v)
      - \frac{1}{\overline{m}} \sum_{\{u, v\} \in E} \dist(u, v)\\
    &= \frac{m + \overline{m}}{\overline{m}}
      \dist(E \cup \overline{E}) - \frac{m}{\overline{m}}\\
    &= \dist(E \cup \overline{E}) +
      \frac{m}{\overline m} (\dist(E \cup \overline{E}) - 1).
  \end{align*}
\end{proof}

\section{Approximating Average Distances}
\label{sec:average-distance}

To approximate the average distance of a graph, we implemented the
algorithm by Chechik, Cohen, and
Kaplan~\cite{Avera_Dista_Queri_throu_APPROX2015}; with a small
improvement that is discussed in
Section~\ref{sec:avg-dist-improvement-details}.  For a parameter $k$,
it computes the BFS trees from roughly $k$ vertices.  We used
$k = \num{400}$ for all graphs.

To evaluate how good the approximation is, we selected eight networks
with roughly \SI{10}{k} vertices and ran the algorithm \num{50} times
on each of them; see Table~\ref{tab:avg_dist_comp}.  We measure the
quality of the approximation using the \emph{relative error}, which is
defined as follows.  Let $G$ be a graph with (exact) average distance
$d$ and let $d'$ be the approximation computed by the algorithm.  Then
the relative error is $|d - d'| / d$.  Table~\ref{tab:avg_dist_comp}
shows that the relative errors are low.  Also note that the networks
were selected to have varying values of locality and heterogeneity.
Specifically, the first two networks are local and homogeneous, the
next two are local and heterogeneous, the fifth and sixth are
non-local and homogeneous, and the last two are non-local and
heterogeneous.

\begin{table}
  \caption{The result of \num{50} runs of approximating the average
    distance with weighted sampling using $k = 400$ samples (in
    expectation) in each run.  The columns are the number of vertices
    ($n$) and edges ($m$), the locality (loc) and heterogeneity (het),
    the exact average distance (avg dist), the smallest (min) and
    largest (max) estimated average distance among the \num{50} runs,
    and the median relative error (error) over the \num{50} runs.}
  \label{tab:avg_dist_comp}%
  \centering%
  \setlength{\tabcolsep}{0.9\tabcolsep}%
  \footnotesize%
  % latex table generated in R 4.2.0 by xtable 1.8-4 package
% Fri Jun 16 13:54:39 2023
\begin{tabular}{lrrrrrrrr}
  \toprule
graph & $n$ & $m$ & loc & het & avg dist & min & max & error \\ 
  \midrule
crack & \num{10240} & \num{30380} & \num{0.77} & \num{-0.50} & \num{41.00} & \num{40.48} & \num{41.68} & \SI{0.44}{\%} \\ 
  inlet & \num{11730} & \num{220296} & \num{0.82} & \num{-0.56} & \num{33.81} & \num{33.24} & \num{34.60} & \SI{0.38}{\%} \\ 
  socfb-Columbia2 & \num{11706} & \num{444295} & \num{0.61} & \num{0.05} & \num{2.84} & \num{2.81} & \num{2.86} & \SI{0.33}{\%} \\ 
  ca-HepPh & \num{11203} & \num{117618} & \num{0.90} & \num{0.36} & \num{4.67} & \num{4.63} & \num{4.71} & \SI{0.21}{\%} \\ 
  sinc15 & \num{11532} & \num{564607} & \num{0.10} & \num{-0.38} & \num{3.05} & \num{3.00} & \num{3.08} & \SI{0.43}{\%} \\ 
  fd15 & \num{11532} & \num{44206} & \num{0.30} & \num{-0.58} & \num{6.02} & \num{5.96} & \num{6.10} & \SI{0.37}{\%} \\ 
  escorts & \num{10106} & \num{39016} & \num{0.24} & \num{0.26} & \num{4.20} & \num{4.18} & \num{4.24} & \SI{0.29}{\%} \\ 
  air03 & \num{10757} & \num{91006} & \num{0.13} & \num{0.83} & \num{2.46} & \num{2.44} & \num{2.48} & \SI{0.21}{\%} \\ 
   \bottomrule
\end{tabular}
\end{table}

In the remainder of this section, we want to discuss two aspects of
computing the average distances that are besides the point of this
paper but nonetheless interesting in their own right.  First, there is
a simple way to improve the approximation quality of the algorithm by
Chechik, Cohen, and Kaplan~\cite{Avera_Dista_Queri_throu_APPROX2015},
which we describe in Section~\ref{sec:avg-dist-improvement-details}.
Moreover, in Section~\ref{sec:avg-dist-uniform-vs-weighted-details},
we compare the algorithm to just sampling vertex-pairs uniformly.

\subsection{Improved Approximation Quality}
\label{sec:avg-dist-improvement-details}

Given a graph $G = (V, E)$ and a parameter $k$, the algorithm by
Chechik, Cohen, and Kaplan~\cite{Avera_Dista_Queri_throu_APPROX2015}
works roughly as follows.  It first computes a probability $p_v$ for
each vertex $v \in V$ such that the sum of all probabilities is in
$O(k)$.  Then, a sample $S \subseteq V$ is created by including $v$ in
$S$ with probability $p_v$, independently of the other vertices.  Note
that $S$ has expected size $\EX{|S|} \in O(k)$.  Finally, for each
vertex $u \in S$, a BFS from $u$ is run, summing the distances from
$u$ to other vertices, scaled with the factor $1/p_u$ to accommodate
for the fact that $u$ was chosen as sample with probability $p_u$.

Our improvement is the following.  Once we have sampled the set $S$,
we know how large $S$ actually is.  Thus, we can condition on the size
of $S$ when looking at the probability that $u$ was chosen as a
sample.  Formally, after sampling $S$, we set
$p_u' = p_u \cdot |S| / \EX{|S|}$, which is the probability that
$u \in S$ conditioned on the size of $S$.  Then, when summing the
distances from $u$ to other vertices, we scale these distances with
$1 / p_u'$ instead of with $1 / p_u$.

To evaluate this change, we ran both variants of the algorithm (with
and without conditioning on the size of $|S|$) for different values of
$k$ on the instances in Table~\ref{tab:avg_dist_comp}.  We ran each
configuration \num{50} times.  The resulting relative errors are shown
in Figure~\ref{fig:avg_dist_comp_correction}.  Note that the our
improved variant conditioning on the size of $S$ yields a substantially
better approximation.

\begin{figure}
  \centering
  \includegraphics{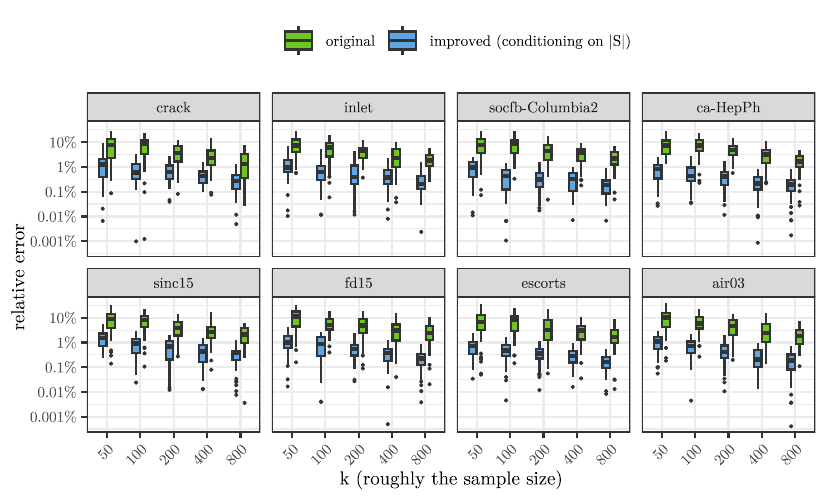}
  \caption{Relative error depending on the parameter $k$, which
    roughly corresponds to the number of samples $|S|$.  Each box
    corresponds to \num{50} runs.  Green boxes show errors for the
    original algorithm as proposed by Chechik, Cohen, and
    Kaplan~\cite{Avera_Dista_Queri_throu_APPROX2015}.  Blue boxes show
    errors after our adjustment of conditioning on $|S|$.  Note that
    both axes are logarithmic.}
  \label{fig:avg_dist_comp_correction}
\end{figure}

\subsection{Uniform vs. Weighted Sampling}
\label{sec:avg-dist-uniform-vs-weighted-details}

In this section, we want to compare the above algorithm with the most
straight-forward method of approximating the average distance:
averaging over the distance between $k$ uniformly and independently
sampled vertex pairs.  In the following, we refer to this as
\emph{uniform sampling}.  Moreover, we refer to the algorithm from
\cite{Avera_Dista_Queri_throu_APPROX2015} with our improvement
described in the previous section as \emph{weighted sampling}.
Figure~\ref{fig:avg_dist_comp} shows a comparison of the average
errors of uniform and weighted sampling depending on the parameter
$k$\punctuationfootnote{Note that for the weighted sampling, the
  expected size of $S$ is in $O(k)$ but not exactly~$k$.  In our
  experiments, $|S|$ was on average $\num{1.19}\cdot k$.  Uniform
  sampling uses exactly $k$ samples.}.  One can see that the weighted
sampling performs better, which is the expected outcome.

\begin{figure}
  \centering
  \includegraphics{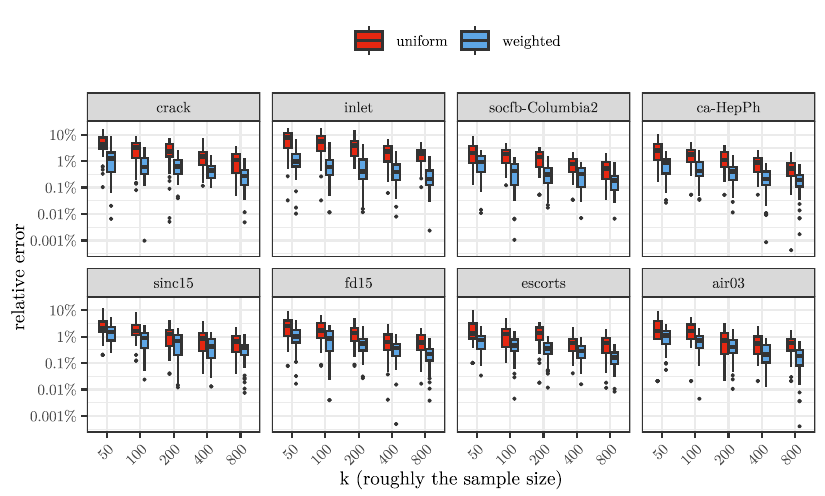}
  \caption{Relative errors depending on the parameter $k$.  Each box
    corresponds to \num{50} runs.  Red and blue boxes show errors for
    uniform and weighted sampling, respectively.  Note that both axes
    are logarithmic.}
  \label{fig:avg_dist_comp}
\end{figure}

\begin{figure}
  \centering
  \includegraphics{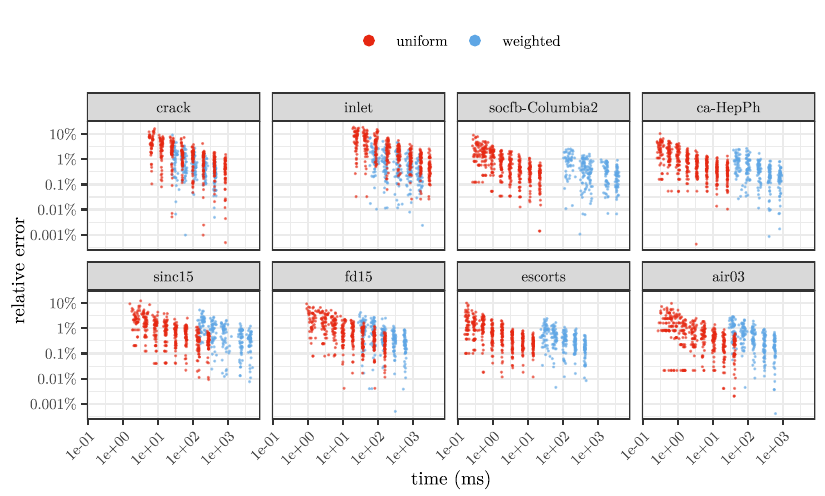}
  \caption{Relative errors depending on the run time.  The parameter
    $k$ determining the sample size ranges from \num{50} to \num{800}
    for the weighted sampling and from \num{50} to \num{6400} for the
    uniform sampling.  The plot shows \num{50} runs for each value of
    $k$.  Note that both axes are logarithmic.}
  \label{fig:avg_dist_comp_time}
\end{figure}

However, this has to be taken with a grain of salt, as weighted
sampling runs a full BFS for every sampled vertex, while uniform
sampling only computes the shortest path between the sampled vertex
pairs.  While both take linear time per sample in the worst case, we
know that the bidirectional BFS can compute shortest paths
substantially quicker on many networks; see
Section~\ref{sec:bidirectional-search}.  To make a more fair
comparison, we additionally ran the uniform sampling for up to
$k = \num{6400}$ vertex pairs and compare the relative errors of
uniform and weighted sampling with respect to running time in
Figure~\ref{fig:avg_dist_comp_time}.  Given our knowledge from
Section~\ref{sec:bidirectional-search}, it is not surprising to see
that weighted sampling is preferable over uniform sampling on the two
networks with high locality and low heterogeneity, as the
bidirectional search is not much faster than running a full BFS for
these networks.  For the other networks, however, doing the weighted
sampling for more samples yields comparable, if not better,
approximations in less time.

This suggests the conclusion that one should use the weighted sampling
on graphs that are local and uniform, while the uniform sampling is
preferable on all other networks.  To cement this conclusion, more
thorough experiments would be necessary, which is beyond the scope of
this paper.

\section{Technical Details Concerning Random Networks}
\label{sec:random-networks-details}

In Section~\ref{sec:giant-component}, we observe that reducing the
generated graphs to their largest connected components does not change
the graph size by too much.  In Section~\ref{sec:girgs-square} we
discuss how we generate GIRGs with square as ground space, instead of
the usual torus.

\subsection{Giant Component}
\label{sec:giant-component}

The generated graphs are not necessarily connected.  Thus, as we
reduce all networks to their largest connected component, the
resulting graphs do not necessarily have $n = \SI{50}{k}$ vertices.
However, the average degree of \num{10} is sufficiently large, so that
the largest connected component is not too much smaller.  Among the
\num{500} generated GIRGs, the smallest graph has still \num{45771}
vertices, with an average of \SI{49.5}{k} vertices.  For the Chung--Lu
graphs, the minimum is \num{45389} and the average \SI{49.2}{k}.  For
Erdős--Rényi, the minimum is already \num{49994}.  Moreover, the
average degrees of all generated networks are close to \num{10}.

\subsection{GIRGs With Square as Ground Space}
\label{sec:girgs-square}

Recall that the GIRG model uses a torus as ground space.  Note that
the torus is completely symmetric in the sense that every point can be
treated the same.  In contrast, if using a square (or hypercube in
higher dimensions) as ground space, there are special cases for points
closer to the boundary of the square, which complicates theoretical
analysis.  For reasons described in Section~\ref{sec:diameter}, the
ground space makes a difference for the diameter computation.  Here we
describe how we used the GIRG generator, which usually works with the
torus, to generate GIRGs with unit square as ground space (still using
the maximum norm).

Recall that the $d$-dimensional torus $\mathbb T^d = [0, 1]^d$ behaves
like the $d$-dimensional cube except that distances wrap around the
boundaries in each dimension.  Thus, when restricting $T^d$ to points
in $[0, 0.5]^d$, the distances within this part of the torus behave
exactly as distances in a $d$-dimensional cube.  We use this to
generate GIRGs with a $d$-dimensional hypercube as follows.  The
generator first samples a point in $[0, 1]^d$ for each vertex.  Before
sampling the edges based on these vertex positions, we scale all
coordinates by $0.5$, yielding points in $[0, 0.5]^d$.  To accommodate
for the fact that this leads to smaller distances, we scale the vertex
weights by $0.5^d$.  Afterwards, we generate the edges as for the
torus~\cite{Effic_Gener_Geome_Inhom_ESA2019}.

To see that scaling all weights by $0.5^d$ achieves roughly the right
average degree, recall that the probability for two vertices $u$ and
$v$ to be connected is
\begin{equation*}
  p_{u, v} =
  \min\left\{\left(\frac{1}{\geomdist{\pnt{u}}{\pnt{v}}^d}\cdot\frac{w_u
        w_v}{W}\right)^{\frac{1}{T}}, 1\right\}.
\end{equation*}
Note that scaling all weights by $0.5^d$ increases $w_uw_v$ by
$0.5^{2d}$ and the sum of all weights $W$ by $0.5^d$, which yields a
total increase of $0.5^d$ contributed by the weights.  Moreover,
scaling all coordinates by $0.5$ decreases most distances by a factor
of $0.5$, which cancels out with the $0.5^d$ coming from the weights.
Note that not all distances are actually scaled by $0.5$: If the
geodesic between $\pnt{u}$ and $\pnt{v}$ wraps around the torus, then
the distance between them might actually be increased due to the
scaling.  However, in the context of GIRGs, this is only relevant for
a sublinear fraction of vertex pairs and is thus not too relevant.
Moreover, by just scaling the weights we also ignore the minimum in
the formula of $p_{u, v}$.  However, the deviation from the desired
average degree was not too big, as we report in the following.

Among the \num{500} generated GIRGs with square as ground space, the
smallest graph still has \num{45468} vertices, with an average of
\SI{49.3}{k} vertices.  The average degrees are slightly below the
target value, ranging from \num{8.99} to \num{10.24} with an average
of \num{9.56}.

\section{Maximal Cliques, Degeneracy, and (Weak) Closure}
\label{sec:maximal-cliques-ext}

In this section we provide additional experiments on the number of
maximal cliques in the context of degeneracy and (weak) closure.
Section~\ref{sec:maximal-cliques-degen-closure} gives a formal
definition of these parameters.  In Section~\ref{sec:comp-weak-clos}
we provide an algorithm for efficiently computing the weak closure,
which might be interesting in its own right.  In
Section~\ref{sec:maximal-cliques-impact-degen-closure} we discuss the
dependence of the number of cliques on the parameters.  In
Section~\ref{sec:maximal-cliques-relation} we study how the parameters
relate to each other.

\subsection{Degeneracy and (Weak) Closure}
\label{sec:maximal-cliques-degen-closure}

The \emph{degeneracy} $d$ of a network is the smallest number such
that iteratively removing vertices of degree at most $d$ eliminates
the whole graph. 

A network is \emph{$c$-closed} if every pair of vertices with at least
$c$ common neighbors is connected.  Note that a single pair of
vertices with many common neighbors already leads to a high
closure. The weak closure addresses this issue as follows.  A pair of
non-adjacent vertices is a \emph{$c$-bad pair} if they have at least
$c$ common neighbors.  We call a vertex \emph{$c$-bad} if it appears
in a $c$-bad pair.  Otherwise, we call it \emph{$c$-good}.  A graph is
\emph{weakly $c$-closed} if iteratively removing $c$-good vertices
eliminates the whole graph.  Note that a $c$-closed graph is also
weakly $c$-closed.

Also note the similarity between weak closure and degeneracy.  Both
are defined via an elimination order on the vertices.  Moreover, a
vertex $v$ can have at most $\deg(v)$ common neighbors with another
vertex.  Thus, for a weakly $c$-closed graph with degeneracy $d$ it
holds that $c - 1 \le d$.

\subsection{Computing the Weak Closure}
\label{sec:comp-weak-clos}

Before we describe how we compute the weak closure, note that the
closure can be computed as follows.  For every vertex $v$, look at the
first two layers of the BFS-tree from $v$ and for each node $w$ in the
second layer, count the number of length-2 paths from $v$ to $w$.  The
node $v$ is $c$-good if this count is strictly lower than $c$ for all
vertices $w$ in the second layer.  Moreover, the graph is $c$-closed
if all vertices are $c$-good.  The running time for each vertex is
dominated by the sum of degrees of vertices in the first layer of the
BFS-tree.  Overall, every vertex $v$ appears $\deg(v)$ times in the
first layer, yielding a running time of
$O\left(\sum_{v\in V}\deg^2(v)\right)$.  This is sufficiently fast for
all networks in our data set.

A naive approach to compute the weak closure is as follows.  Start
with $c = 1$.  Then run the above procedure to find the minimum value
$c'$ such that there exists a $c'$-good vertex.  If $c < c'$, set
$c = c'$ and then remove all $c$-good vertices.  Repeat this until the
graph is completely eliminated.  The graph is then weakly $c$-closed
for the final value of $c$.  Though this is fast for many networks, it
is prohibitively slow for some.
In the following, we describe how to improve the running time at the
cost of a higher memory consumption.  Afterwards, we discuss how to
decrease the memory footprint to a reasonable level.

For every vertex $v$, we have a priority queue $Q_v$ that contains all
vertices of distance two from $v$.  For such a vertex $w$, the
priority is set to the number of common neighbors of $v$ and $w$,
i.e., a priority of $c_{v, w}$ indicates that $v$ and $w$ form a
$c_{v, w}$-bad pair.  Let $c_v$ be the largest priority in $Q_v$.
Note that, $v$ is $c_v$-bad but $(c_v + 1)$-good.  Thus, we want to
iteratively remove the vertex $v$ with minimum $c_v$.  To this end, we
maintain one additional priority queue $Q$ containing all vertices,
using $c_v$ as priority for $v$.

When removing a vertex $u$, we have to update the neighbors of $u$, as
they lose $u$ as common neighbor.  This is done as follows.  Let $v$
and $w$ be two neighbors of $u$ that are not connected by an edge.
Recall that the queue $Q_v$ contains $w$ with the priority $c_{v, w}$
indicating the number of common neighbors of $v$ and $w$.  As they
lost $u$ as common neighbor, we have to decrease this priority in
$Q_v$ by \num{1}.  If this decreases $c_v$, i.e., the maximum priority
in $Q_v$, we also have to adapt the priority of $v$ in the queue $Q$
accordingly.

This algorithm takes $O\left(\sum_{v\in V}\deg^2(v)\right)$ time to
initialize the data structures.  Moreover, when removing vertex $u$,
we have to update $\deg^2(u)$ neighbor pairs.  Assuming all
queue-operations take constant time,\footnote{This can be achieved
  using a bucket heap that makes use of the fact that the range of
  possible integer priorities is bounded.} this takes overall
$O\left(\sum_{v\in V}\deg^2(v)\right)$ time.  Unfortunately, it also
requires that amount of memory, which is prohibitive for some
instances.

To improve the memory consumption, we make use of the following
observation.  Taking the squares of the degrees is particularly bad if
the graph contains vertices of high degree.  A vertex of high degree
is responsible for many pairs of vertices that have distance \num{2}.
However, most of these pairs do not have many common neighbors, and
thus will never be a $c$-bad pair for a relevant value of $c$.  To
phrase it differently, if we want to show that the graph is weakly
$c$-closed, we can ignore all distance-2 pairs $v, w \in V$ that have
fewer than $c$ common neighbors.  Thus, using the above notation for
the queues, we do not have to insert $w$ in the queue $Q_v$ if
$c_{v, w} < c$.

Thus, we start with the guess that the given graph is $\bar c$-closed
for some value $\bar c$.  Then, we compute the elimination order as
above, but in the initialization of the data structures, we ignore all
pairs with fewer than $\bar c$ common neighbors.  Assume the procedure
concludes that the graph is $c$-closed, i.e., all vertices were
$c$-good at the time of removal, but some $(c-1)$-bad vertices had to be
removed.  If $c \ge \bar c$, then the ignored vertex pairs have fewer
than $c$ common neighbors.  Thus, none of these vertex pairs turns a
$c$-good vertex into a $c$-bad vertex, implying that the result is
correct despite ignoring some vertex pairs.  On the other hand, if
$c < \bar c$, we may have ignored some crucial pairs and have to rerun
the procedure with lower $\bar c$.

Preliminary experiments showed that guessing $\bar c$ even slightly
too low can lead to high memory consumption, while guessing $\bar c$
too high is computationally not very expensive.  Starting with
$\bar c = 30$ and successively decreasing it by \num{1} if necessary
yields acceptable\footnote{We note that there is probably still plenty
  room for fine-tuning.  However, this is beyond the scope of this
  paper.}  running times and memory footprints for all instances in
our data set.

\begin{figure}[t!]
  \centering
  \includegraphics{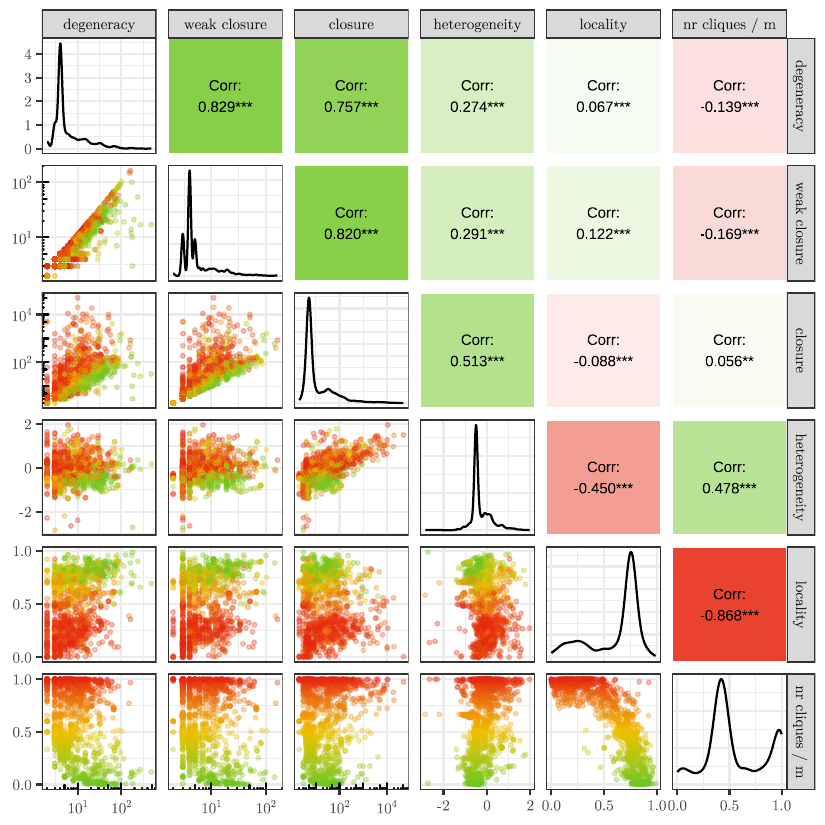}
  \caption{Pairwise comparison of different variables for the networks
    with at most $m$ cliques.  The upper triangle shows the Spearman
    correlation between the two corresponding variables. The colors go
    from green (positive correlation) over white (no correlation) to
    red (negative correlation).  The stars indicate p-values
    (``\texttt{***}'': $< 0.001$, ``\texttt{**}'': $< 0.01$,
    ``\texttt{*}'': $< 0.05$, ``\texttt{.}'': $< 0.1$, ``~'':
    otherwise).  The diagonal shows the density of each individual
    variable.  The lower triangle shows scatter plots of the networks
    with respect to two variables.  The colors indicate the number of
    cliques (relative to $m$).  Axes for degeneracy, weak closure,
    and closure are logarithmic.}
  \label{fig:cliques-corr}
\end{figure}

\begin{figure}[t!]
  \centering
  \includegraphics{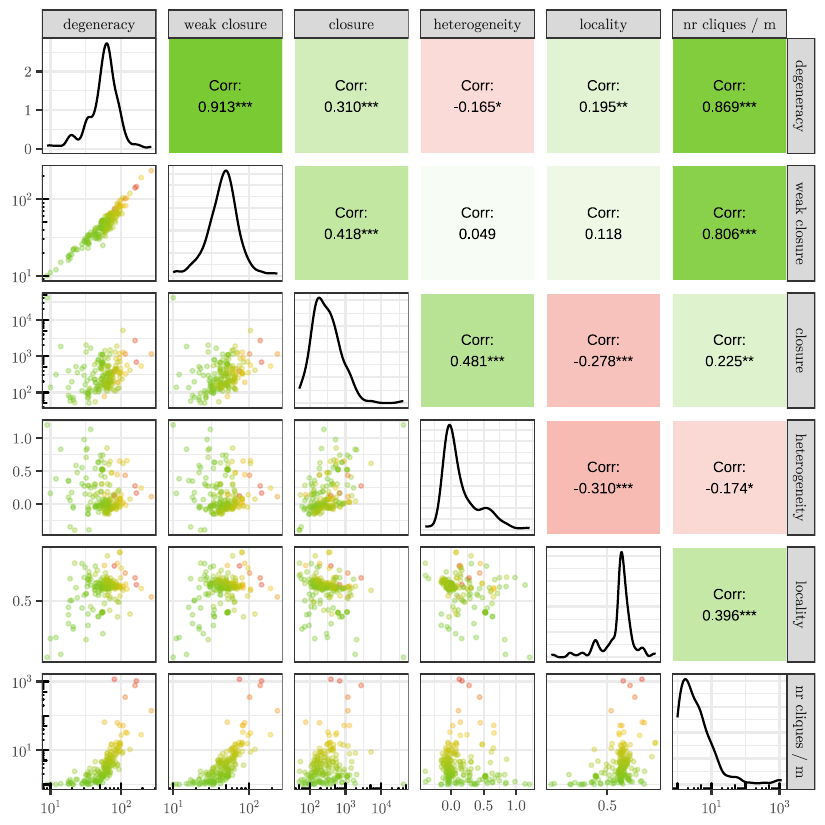}
  \caption{Pairwise comparison as in Figure~\ref{fig:cliques-corr} but
    for the networks with more than $m$ cliques.  In addition to
    degeneracy, weak closure, and closure, the axis for the relative
    number of cliques is logarithmic.}
  \label{fig:cliques-corr-outlier}
\end{figure}

\subsection{Impact of Degeneracy and Closure on the Number of Cliques}
\label{sec:maximal-cliques-impact-degen-closure}

Figure~\ref{fig:cliques-corr} and
Figure~\ref{fig:cliques-corr-outlier} show a pairwise comparison of
degeneracy, weak closure, closure, heterogeneity, locality, and number
of cliques (relative to $m$) for all networks with at most $m$ cliques
and more than $m$ cliques, respectively.  The bottom row and the
right-most column compare the number of cliques with the other
parameters.

\subsubsection{Networks With at Most $m$ Cliques.}

In Figure~\ref{fig:maximal-cliques}, we already saw that the number of
cliques increases for decreasing locality and we saw a slight increase
for increasing heterogeneity.  This matches to what we observe here: a
strong negative correlation to the locality and a weaker positive
correlation to the heterogeneity.

For degeneracy and (weak) closure, theoretical results show that low
values for these parameters guarantee a low number of
cliques~\cite{Listi_Maxim_Cliqu_Large_jour2013,
  Findi_Cliqu_Socia_Netwo_jour2020}.  Though these theoretical bounds
operate in a completely different regime (above $m$ by factors
exponential in the parameter, not below $m$), one could nonetheless
hope that these parameters serve as a good measure for the hardness of
an instance, i.e., that the number of cliques positively correlates
with them.  Figure~\ref{fig:cliques-corr} shows that this hope is not
justified.  There is little to no correlation with the closure and
even a slightly negative correlation with degeneracy and weak closure.

\subsubsection{Networks With More Than $m$ Cliques.}

For the networks with more than $m$ cliques
(Figure~\ref{fig:cliques-corr-outlier}), one can see a strong positive
correlation of the number of cliques with the degeneracy and the weak
closure.  Thus, for these fewer somewhat hard instances, the
degeneracy and weak closure serve as good measures for how hard an
instance actually is.  This qualitatively matches the theoretical
bounds of $O(dn3^{d/3})$ \cite{Listi_Maxim_Cliqu_Large_jour2013} and
$n^23^{(c - 1)/3}$ \cite{Findi_Cliqu_Socia_Netwo_jour2020}, where $d$
is the degeneracy and $c$ the weak closure.  For the closure on the
other hand, there is only a slight correlation.

\subsection{Relation Between the Parameters}
\label{sec:maximal-cliques-relation}

Besides studying the number of cliques with respect to the different
parameters, it is also interesting to compare the parameters with each
other.  Closure and weak closure have both been introduced to
formalize the same concept (triadic closure), thus one would suspect
them to be similar.  Moreover, recall that the weak closure and the
degeneracy are both defined via elimination orders on the vertices and
that $c - 1 \le d$ for weakly $c$-closed graphs with degeneracy $d$.
Thus, one can expect them to be similar on networks with few
triangles.  For graphs with many triangles (high locality), it is
interesting to see whether the weak closure captures the concept of
triadic closure well, i.e., whether $c - 1$ is substantially smaller
than $d$.

Though the study of the relation between these parameters is
independent of the number of cliques, we still consider the partition
into networks with at most or more than $m$ cliques for the following
reason.  Our interest in degeneracy and (weak) closure comes from
trying to understand how many cliques a network has.  As we have seen
before, these parameters do not work well in this regard for networks
with at most $m$ cliques, i.e., they are more relevant for the
networks with more than $m$ cliques.  If we, however, would consider
all networks together, the networks with at most $m$ cliques would
dominate the overall picture as they make up \SI{93}{\%} of all
networks.

\subsubsection{Networks With at Most $m$ Cliques.}

We can see in Figure~\ref{fig:cliques-corr} that degeneracy, closure,
and weak closure are all positively correlated.  We can also see that
the degeneracy and the weak closure range in the same order of
magnitude, while the closure is orders of magnitude larger.

Concerning correlation to heterogeneity and locality, we can observe a
slight positive correlation of all three with the heterogeneity.  This
makes sense as all three parameters can only be high if there are
vertices of high degree.  Moreover, the correlation with locality is
weaker or not present at all.  For the (weak) closure, this is
unexpected as they are meant to capture the existence of many
triangles, which corresponds to a high locality.  It is particularly
surprising that there is a slight positive correlation of the weak
closure with locality, i.e., a higher locality tends to lead to a
larger weak closure, which is the opposite of what one would expect.

A possible explanation for this can be obtained by observing that
closure focuses on non-edges while our definition of locality focuses
on edges.  Slightly simplifying, this difference can be stated as
follows.
\begin{description}
\item[high locality:] $\{u, v\} \in E$ $\rightarrow$ $u$ and $v$ have
  many common neighbors
\item[small closure:] $\{u, v\} \notin E$ $\rightarrow$ $u$ and $v$ do
  not have many common neighbors
\end{description}

\begin{figure}[t]
  \centering
  \begin{minipage}[t]{.48\textwidth}
    \centering
    \includegraphics{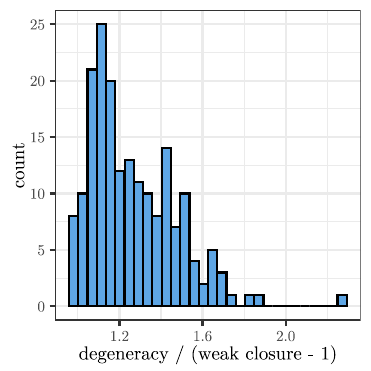}
    \captionof{figure}{Distribution of the relative difference between
      degeneracy and weak closure values for the networks with more
      than $m$ cliques.}
    \label{fig:degen_vs_weak_closure}
  \end{minipage}\hfill
  \begin{minipage}[t]{.48\textwidth}
    \centering
    \includegraphics{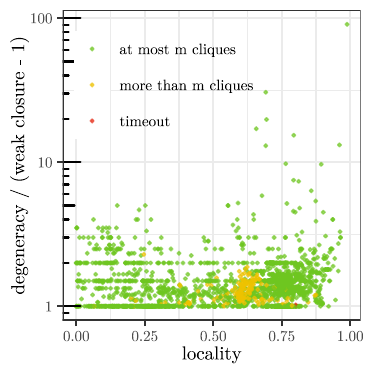}
    \captionof{figure}{Relative difference between degeneracy and weak
      closure depending on the locality.  They show only a slight
      correlation (\num{0.10} for networks with at most $m$ cliques
      and \num{0.24} for networks with more than $m$ cliques).}
    \label{fig:degen_vs_weak_closure_locality}
  \end{minipage}
\end{figure}

\subsubsection{Networks With More Than $m$ Cliques.}

For the networks where the degeneracy and weak closure were a good
predictor for the number of cliques, we see an even stronger
correlation between these two parameters.  To quantify how much the
degeneracy and the weak closure differ, see
Figure~\ref{fig:degen_vs_weak_closure}.  One can see that the weak
closure is usually not much smaller than the upper bound given by the
degeneracy.  Thus, the weak closure is indeed very similar to the
degeneracy.  Hence it mostly captures the sparsity of a network rather
than the tendency to have many triangles.

On these networks the correlation of the closure to degeneracy and
weak closure is less pronounced than for the networks with at most $m$
cliques.  Concerning locality and heterogeneity, there is no
correlation with degeneracy or weak closure.  The closure correlates
positively with heterogeneity and to a smaller extent negatively with
locality.  As for the other set of networks, this again indicates that
the parameter closure is more susceptible to the degree distribution
than to the existence of triangles.

\subsubsection{Weak Closure and Locality}

As mentioned earlier, the correlation between weak closure and
degeneracy is not surprising as degeneracy is an upper bound to weak
closure (minus \num{1}).  Here we want to study whether the difference
between these two is correlated with the locality.  If this difference
comes mostly from the existence of many triangles, then one would
expect a bigger difference for graphs with high locality.  However,
Figure~\ref{fig:degen_vs_weak_closure_locality} shows that this is not
really the case.  This consolidates the previous observation that
(weak) closure does not capture the concept of locality well.

\section{Networks with Extreme Heterogeneity}
\label{sec:extreme-heterogeneity-ext}

In this section, we provide the figures from Section~\ref{sec:comp-betw-real-world-and-models} including real-world networks with extreme heterogeneity. Figure~\ref{fig:extreme-heterogeneity-1} and Figure~\ref{fig:extreme-heterogeneity-2} show the respective middle plots; the thresholds for extreme heterogeneity are marked with vertical lines, with triangle-shaped data points representing real-world networks with extreme heterogeneity, i.e., data points outside the thresholds.
Across all considered algorithms, the behavior on networks with extreme heterogeneity roughly follows those of the other networks, showing similar trends of algorithm behavior dependence on locality and heterogeneity.

\begin{figure}[t!]
	\centering
	\begin{subfigure}{0.45\textwidth}
		\centering
		\includegraphics{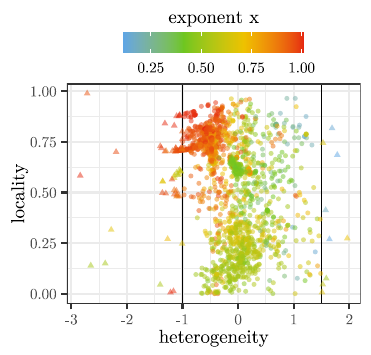}
		\caption{Figure~\ref{fig:bbbfs} (bidirectional BFS)}
		\label{fig:bbbfs-full}
	\end{subfigure}
	\begin{subfigure}{0.45\textwidth}
		\centering
		\includegraphics{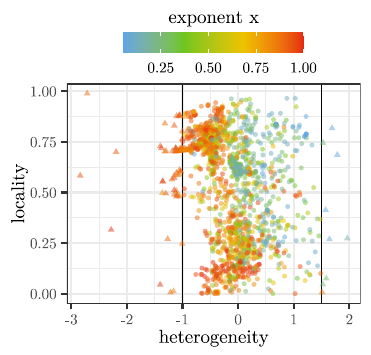}
		\caption{Figure~\ref{fig:diameter-ifub-hd} (iFUB+hd)}
		\label{fig:diameter-ifub-hd-full}
	\end{subfigure}
	\begin{subfigure}{0.45\textwidth}
		\centering
		\includegraphics{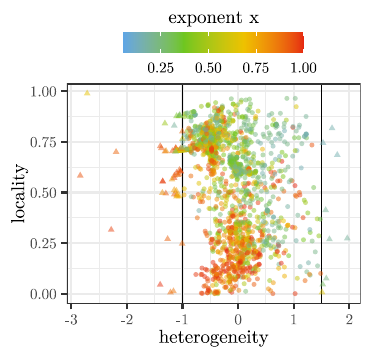}
		\caption{Figure~\ref{fig:diameter-ifub-foursweep} (iFUB+4-sweephd)}
		\label{fig:diameter-ifub-foursweephd-full}
	\end{subfigure}
	\begin{subfigure}{0.45\textwidth}
		\centering
		\includegraphics{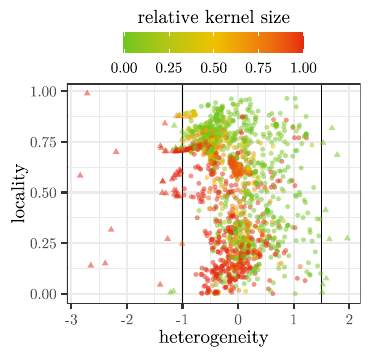}
		\caption{Figure~\ref{fig:vertex-cover-domination} (vertex cover domination)}
		\label{fig:vertex-cover-full}
	\end{subfigure}
	\begin{subfigure}{0.45\textwidth}
		\centering
		\includegraphics{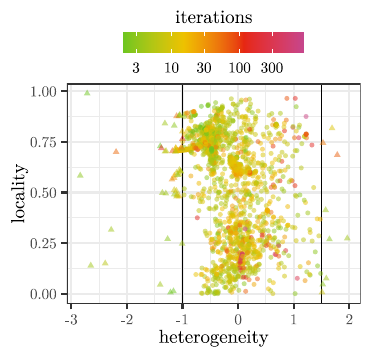}
		\caption{Figure~\ref{fig:louvain} (Louvain)}
		\label{fig:louvain-full}
	\end{subfigure}
	\caption{Full versions of the middle plots of several figures from Section~\ref{sec:comp-betw-real-world-and-models} including real-world networks with extreme heterogeneity. The thresholds for extreme heterogeneity are marked with vertical lines, with triangle-shaped data points representing real-world networks with extreme heterogeneity.}
	\label{fig:extreme-heterogeneity-1}
\end{figure}

\begin{figure}[t!]
	\centering
	\begin{subfigure}{0.45\textwidth}
		\centering
		\includegraphics{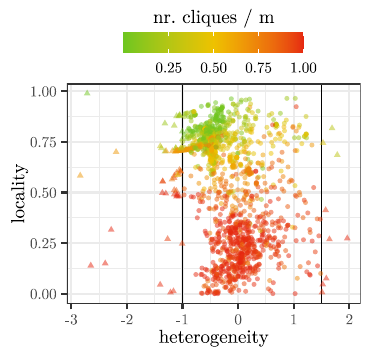}
		\caption{Figure~\ref{fig:maximal-cliques} (maximal cliques)}
		\label{fig:cliques-full}
	\end{subfigure}
	\begin{subfigure}{0.45\textwidth}
		\centering
		\includegraphics{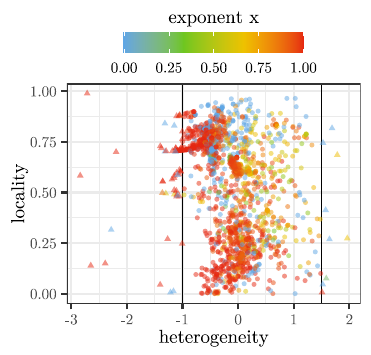}
		\caption{Figure~\ref{fig:coloring-low-deg} (chromatic number reduction)}
		\label{fig:coloring-full}
	\end{subfigure}
	\begin{subfigure}{0.45\textwidth}
		\centering
		\includegraphics{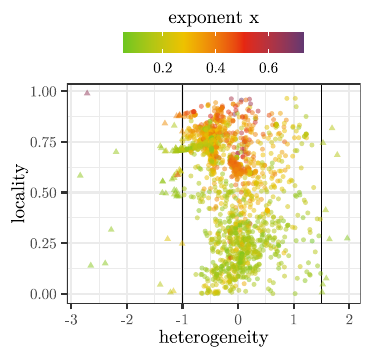}
		\caption{Figure~\ref{fig:max-clique-size-and-degen-clique} (clique number)}
		\label{fig:cliques-max-clique-full}
	\end{subfigure}
	\begin{subfigure}{0.45\textwidth}
		\centering
		\includegraphics{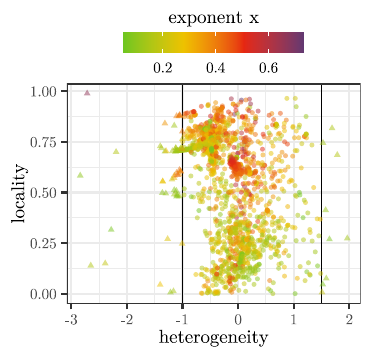}
		\caption{Figure~\ref{fig:max-clique-size-and-degen-degen} (degeneracy)}
		\label{fig:degeneracy-full}
	\end{subfigure}
	\caption{Additional full versions of the middle plots of several figures from Section~\ref{sec:comp-betw-real-world-and-models} including real-world networks with extreme heterogeneity. The thresholds for extreme heterogeneity are marked with vertical lines, with triangle-shaped data points representing real-world networks with extreme heterogeneity.}
	\label{fig:extreme-heterogeneity-2}
\end{figure}

\end{document}